\documentclass[%
superscriptaddress,
nofootinbib,
 amsmath,amssymb,
 aps,
onecolumn,
showkeys
]{revtex4-2}

\usepackage{todonotes}
\usepackage{graphicx}
\usepackage{dcolumn}
\usepackage{bm}
\usepackage{hyperref}

\usepackage[
marginparwidth=1in,
]{geometry}

\usepackage{amsthm}
\usepackage{stmaryrd}
\usepackage[export]{adjustbox}

\newtheorem{theorem}{Theorem}[section]

\newtheorem{lemma}{Lemma}
\newtheorem{prop}{Proposition}

\newcommand{\cG}{\mathcal{G}}
\newcommand{\cC}{\mathcal{C}}
\newcommand{\cS}{\mathcal{S}}
\newcommand{\cM}{\mathcal{M}}

\DeclareMathOperator{\Tr}{Tr}

\DeclareMathOperator{\Cat}{Cat}

\begin{document}


\title{Double scaling limit of multi-matrix models at large $D$}

\author{{\bf V. Bonzom}}\email{bonzom@lipn.univ-paris13.fr}
\affiliation{Universit\'e Sorbonne Paris Nord, LIPN, CNRS UMR 7030, F-93430 Villetaneuse, France, EU}

\author{{\bf V. Nador}}\email{victor.nador@u-bordeaux.fr}
\affiliation{LaBRI, Univ. Bordeaux, 351 cours de la Lib\'eration, 33405 Talence, France, EU}

\author{{\bf A. Tanasa}}\email{ntanasa@u-bordeaux.fr}
\affiliation{LaBRI, Univ. Bordeaux, 351 cours de la Lib\'eration, 33405 Talence, France, EU}
\affiliation{H. Hulubei Nat. Inst. Phys. Nucl. Engineering,P.O.Box MG-6, 077125 Magurele, Romania, EU}

\date{\today}

\begin{abstract}
In this paper, we study a double scaling limit of two multi-matrix models: the $U(N)^2 \times O(D)$-invariant model with all quartic interactions and the bipartite $U(N) \times O(D)$-invariant model with tetrahedral interaction ($D$ being here the number of matrices and $N$ being the size of each matrix). Those models admit a double, large $N$ and large $D$ expansion. While $N$ tracks the genus of the Feynman graphs, $D$ tracks another quantity called the grade. In both models, we rewrite the sum over Feynman graphs at fixed genus and grade as a finite sum over combinatorial objects called schemes. This is a result of combinatorial nature which remains true in the quantum mechanical setting and in quantum field theory. Then we proceed to the double scaling limit at large $D$, i.e. for vanishing grade. In particular, we find that the most singular schemes, in both models, are the same as those found in Benedetti et al. for the $U(N)^2 \times O(D)$-invariant model restricted to its tetrahedral interaction. This is a different universality class than in the 1-matrix model whose double scaling is not summable.
\end{abstract}


\maketitle

\tableofcontents

\section{Introduction}


Multi-matrix models are matrix models involving $D$ matrices $X_{\mu}$ ($\mu=1,\ldots D$) of size $N\times N$, with a continuous symmetry group acting on the matrix index $\mu$, for example $O(D)$ or $U(D)$~\cite{Ferrari}. Usually, the models are taken to be invariant under the action of a second group, which here is taken to be $U(N)^2$ (acting on the left and right of $X_\mu$) or $U(N)$ (acting by conjugation on $X_\mu$). This type of models is a natural generalization of the celebrated one-matrix model~\cite{DFGiZJ,BreIt,Kaz}. 
They are also related to tensor models~\cite{Gu3,GuRy2,Ta1,Ta3} with three indices, in the sense that one of the three indices plays the role of a vector index transforming under $O(D)$ or $U(D)$. This sitting of multi-matrix models at the crossroad of matrix and tensor models makes them specially interesting. 

In addition, they have been related to
the large $D$ limit of Einstein equations and to
black hole physics, see \cite{Emparan1} or the review \cite{Emparan2}.

Instead of the classical large $N$ expansion on matrix models, there are now two parameters, $N$ and $D$ which can be taken large, and some classes of models are known to admit a double, large $N$ and large $D$ expansion (in particular, a large $N$ expansion of tensor models is found by setting $D=N$). However, proving its existence can be a challenging task in itself~\cite{TaFe}.

The Feynman graphs are no longer ordinary ribbon graphs, but can be seen as ribbon graphs decorated with cycles which carry a $D$-dependent weight. The parameter $N$ controls a genus expansion\footnote{In fact, for general interactions, this is not 
a genus anymore, but the sum of the genus and another positive quantity, as we explain in Equation~\eqref{GenusExpansion2}. A more general case was already described in \cite{NewExpansions}.}. The parameter $D$ controls the expansion with respect to another combinatorial quantity called the \emph{grade}, introduced in~\cite{FeVa}. The graphs which dominate the large $N$, large $D$ expansions are those of vanishing genus and grade. They have been shown to have the same structure as the melonic graphs which dominate some large $N$ limits in tensor models~\cite{BoGuRiRi, Dartois, Dario, CarrozzaHarribey2021, CarrozzaPozsgay2018, Carrozza2018}, and are therefore also called \emph{melonic} graphs~\cite{TaCa} (we point out that tensor models with different interactions can have non-melonic large $N$ limits, as shown in \cite{NewExpansions, Bonzom2016, Bonzom2018, StuffedMaps, GeneralizedMelons, Octahedra, Enhancing, Lionni}). 


In this article we go beyond the large $N$, large $D$ limit and proceed to the \emph{scheme classification}. For convenience, we consider 2-point graphs\footnote{For the second model we study, we consider in fact rooted graphs, which are vacuum graphs with a marked edge. They are equivalent to 2-point graphs but their $N$-dependent weight is the same as a vacuum graph, whereas 2-point graphs can have one or two broken faces.}
instead of vacuum graphs, because this allows for ignoring the automorphism group of vacuum graphs. We are able to rewrite the set of (2-point) Feynman graphs of fixed genus and grade as a \emph{finite} set of combinatorial objects called schemes, decorated with chains (also known as ladders in the physics literature) and melons. This result is purely combinatorial and holds in the settings of multi-matrix quantum mechanics and quantum field theory. In terms of generating series (the free energy and 2-point function of the matrix models), this allows for the rewriting, at fixed genus and grade, of the sum over an infinite set of Feynman graphs as a finite sum, more precisely a polynomial of the series of chains and melons.

The scheme classification was originally developed for ribbon graphs (also known as combinatorial maps in mathematics)~\cite{ChapuyMarcusSchaeffer}, where schemes of a given genus are in finite number and all maps of the same genus are recovered by tree decorations. This setup was then redesigned for tensor models by Gurau and Schaeffer \cite{GuSch}. It has since been successfully applied to the multi-orientable model~\cite{TaFu} and the $O(N)^3$-invariant model with quartic interactions~\cite{Bonzom4}.

In the case of one-matrix models, the double scaling mechanism~\cite{DoSh,BreKa,GroMi} is a key mechanism since it is related to the continuum limit of the model. It is obtained by taking $N$ to infinity and the coupling constant of the model $\lambda$ to a critical value while holding a ratio of the two fixed so that the free energy picks up contributions from all genera and is a function of this ration. The double scaling mechanism has been implemented for the tensor models of \cite{GuSch, TaFu, Bonzom4} (see \cite{GuTaYo}) using their scheme classification. Indeed, at a fixed order in the large $N$ expansion, there are finitely many schemes so their sum cannot lead to singularities. Therefore the singularities are solely due to the decorations by chains and melons. The double scaling limit then follows from the identification of the most singular schemes at fixed order in the large $N$ expansion. Interestingly, the double scaling limit in those three cases reduces to a sum over the same type of objects, rooted binary plane trees (see a definition at the end of the introduction) which are in one-to-one correspondence with the most singular schemes.

In \cite{BeCa}, Benedetti et al. studied several scaling limits for the multi-matrix $U(N)^2 \times O(D)$-invariant model with tetrahedral interaction, taking advantage of the new parameter $D$. In particular, they consider the double scaling limit after taking $D$ to infinity, i.e. at vanishing grade. Again, this double scaling limit turns out to be described as a sum over rooted binary plane trees, decorated by chains and melons.


In this paper we perform the scheme classification and the double scaling limit in two models:
\begin{itemize}
    \item In Section \ref{sec:UN2OD}, the $U(N)^2\times O(D)$-invariant model with all quartic interactions. It is known from \cite{FeVa} to admit a large $N$, large $D$ expansion.
    \item In Section \ref{sec:UNODMM}, the complex bipartite $U(N)\times O(D)$-invariant model with a tetrahedral interaction, whose large $N$, large $D$ expansion was established in \cite{TaFe}.
\end{itemize}
The strategy, and the main theorems which follow from it, follows the main general lines as in \cite{GuSch, TaFu, Bonzom4}, changing their large $N$ expansion to the large $N$, large $D$ expansion{\color{red}:}
\begin{enumerate}
    \item \label{enum:Classify} Classify the graphs according to their genus and grade using the scheme decomposition.
    \item \label{enum:WriteGF} Write the generating series of graphs at fixed genus and grade in terms of the known series of chains and melons, and identify their singularities.
    \item \label{enum:DS} Describe the most singular contributions and resum them using a double scaling limit.
\end{enumerate}

The main theorems are the following.
\begin{theorem}
Any 2-point graph can be reconstructed from a unique scheme of the same genus and grade by extending chains and adding melons.
\label{thm:graph-scheme}
\end{theorem}
In other words, each scheme represents an infinite family of graphs. It is therefore possible to repackage the sum of all graphs of any given genus $g$ and grade $l$ as a sum over schemes of the same genus and grade. The 2-point function takes the form
\begin{equation*}
G_{g,l} = \sum_{\text{Schemes $\mathcal{S}$}} P_{\mathcal{S}}(C(M), M)
\end{equation*}
where $P_{\mathcal{S}}$ is a polynomial, $C$ the generating series of chains, and $M$ the generating series of melons. The quantity $P_{\mathcal{S}}(C(M), M)$ is the amplitude resulting from the sum over all graphs associate to the scheme $\mathcal{S}$. The singularities of $G_{g,l}$ may then come from the series $C$, $M$ and from the sum over schemes of genus $g$ and grade $l$ if there is an infinite number of them. This is however not the case.
\begin{theorem}
The set of schemes at fixed genus and grade is finite in both models.
\label{thm:sch}
\end{theorem}
Characterizing all schemes of a given genus and grade is still a hard combinatorial problem which has not been solved (also for tensor models). 

Just like in \cite{BeCa}, we then restrict attention to schemes of vanishing grade, then perform the double scaling limit for which only a subset of schemes, those which are the most singular, contributes.
\begin{theorem} \label{thm:DS}
The double scaling limit at large $D$ is dominated by rooted binary plane trees.
\end{theorem}
We recall that a binary tree is a tree (i.e. a graph with no cycle), such that each vertex has exactly zero or two children. Vertices with no children are called leaves. In a plane tree, the tree is embedded so that the two children of a vertex can be labeled as left and right.

Let us recall that, with respect to \cite{BeCa}, in our study of the $U(N)^2 \times O(D)$-invariant model, we will allow for all quartic interactions (instead of just the tetrahedral one). We perform the double scaling limit at large $D$ in a similar manner, and in addition prove the more general result that there is a finite number of schemes at fixed genus and grade. We further expect that the limit also studied in \cite{BeCa} is analogous even when introducing all quartic interactions.

As we have already pointed out, our strategy and results are similar to previous scheme decompositions and double scaling limits in the literature. To further emphasize the universality of the approach we have organized both Sections \ref{sec:UN2OD} and \ref{sec:UNODMM} similarly. However, while we could have followed again the same proofs as in \cite{Bonzom4}, we have decided to offer variations which we found interesting. One variation is due to the fact that we consider multi-matrix models, which can be seen both as matrix models and tensor models. In the model of Section \ref{sec:UNODMM}, it is convenient to use exclusively the ribbon graph representation inherited from its matrix aspect. In the model of Section \ref{sec:UN2OD} however, we find it convenient to change the representation. We offer a new proof of the large $N$, large $D$ expansion based on the ribbon graph picture, but the rest of the analysis is conducted using the representation as edge-colored graphs, inherited from the tensor formulation. The representation (ribbon graphs or edge-colored graphs) however does not affect the general strategy of the proofs in both models. 

The most difficult part in both models, as was already the case in \cite{GuSch, TaFu, Bonzom4}, is to show that there is a finite number of schemes at fixed genus and grade. Instead of reproducing the proof given in \cite{Bonzom4}, we have chosen a different approach in both models.
\begin{itemize}
    \item For the $U(N)^2\times O(D)$-invariant model, we show that Theorem \ref{thm:sch} is in fact a corollary of the same result proved in \cite{Bonzom4} for the $O(N)^3$-invariant model. However, the most singular schemes at vanishing grade cannot be derived from those which dominate the double scaling limit of the $O(N)^3$-invariant model. This is due to the large $D$ limit which projects on schemes of vanishing grade. Therefore, one has to proceed to an independent analysis to find the most singular schemes, and we use the same method as in \cite{BeCa} for this. Nevertheless, Theorem \ref{thm:DS} gives the same universality class (that of trees) as for the $O(N)^3$-invariant model.
    \item For the bipartite $U(N)\times O(D)$-invariant model, we prove Theorem \ref{thm:sch} by expanding on a method developed in \cite{TaFu}. It is based on identifying topological minors of non-zero genus. Theorem \ref{thm:DS} also follows from the same analysis as in \cite{BeCa}.
\end{itemize}
The diversity of the possible approaches and the possibility of switching between between a matrix and a tensor point of view make those models particularly interesting at the combinatorial level, in addition to their physics motivations.

\section{\label{sec:UN2OD} The \texorpdfstring{$U(N)^2\times O(D)$}{U(N)2XO(D)} multi-matrix model with quartic interaction}

\subsection{Definition of the model and its large \texorpdfstring{$N$}{N}, large \texorpdfstring{$D$}{D} expansion}

\subsubsection{Feynman graphs, genus and grade}

\paragraph{Invariant polynomials.\\}

The $U(N)^2 \times O(D)$ multi-matrix model is a model involving a vector of $D$ complex matrices of size $N \times N$, denoted $(X_\mu)_{\mu=1, \dotsc, D} = (X_1, \dotsc, X_D)$. The model is required to be invariant under unitary actions on the left and on the right of each $X_\mu$,
\begin{equation}
X_\mu \rightarrow X'_\mu = U_1 X_{\mu} U_2^\dagger
\end{equation}
with $U_1,U_2 \in U(N)$. The ring of polynomials which are invariant under this transformation is generated by products of traces of the form
\begin{equation} \label{U(N)2Invariance}
\Tr X_{\mu_1} X^\dagger_{\nu_1} \dotsm X_{\mu_n} X^\dagger_{\nu_n}.
\end{equation}
The model is further required to be invariant orthogonal transformations on the vector $(X_\mu)_{\mu=1, \dotsc, D}$,
\begin{equation}
X_\mu \rightarrow X'_\mu = \sum_{\mu'=1}^D O_{\mu\mu'}\ X_{\mu'} 
\end{equation}
for any $O \in O(D)$. To enforce this on polynomials which are products of the traces of the type \eqref{U(N)2Invariance}, each vector index $\mu_i$ and $\nu_i$ must be identified with another vector index as follows,
\begin{equation} \label{O(D)Invariance}
\sum_{\mu=1}^D X_\mu \dotsb X_{\mu}\dotsb \quad \text{or} \quad \sum_{\mu=1}^D X_\mu \dotsb X^\dagger_{\mu}\dotsb \quad \text{or} \quad \sum_{\mu=1}^D X^\dagger_\mu \dotsb X^\dagger_{\mu}\dotsb
\end{equation}

Graphical rules can be given to represent invariant polynomials. 
\begin{itemize}
\item We represent $(X_\mu)_{ab}$ as a white vertex, $(X^\dagger_\nu)_{b'a'}$ as a black vertex.
\item Their matrix indices $a, a'$ are half-edges of color 1, and their matrix indices $b, b'$ are half-edges of color 2, and the vector indices $\mu, \nu$ are half-edges of color 3. 
\item One then forms graphs by connecting half-edges of the same color, with the only constraint being that the subgraph obtained by removing all edges of color 3 is bipartite. 
\end{itemize}
The graph corresponding to a polynomial is called the bubble, and it is enough to define the polynomial.
The product of traces like in \eqref{U(N)2Invariance} implies graphically that the subgraph with colors 1 and 2 only is a disjoint union of bipartite cycles whose edges alternate the colors 1 and 2. The identification of vector indices from \eqref{O(D)Invariance} then means that the edges of color 3 form a perfect matching on the vertices (black and white undifferently).
We will only consider connected, quadratic and quartic interactions. Here connected means that the bubble is connected. There is a single quadratic bubble, which has one black and one white vertex, and all edges of colors 1, 2, 3 between them,
\begin{equation}
I_k(X, X^\dagger) = \sum_{\mu=1}^D \Tr X_\mu X^\dagger_\mu.
\end{equation}
Quartic invariants can have either one trace, like
\begin{equation}
\sum_{\mu, \nu, \rho, \sigma} t_{\mu\nu\rho\sigma} \Tr X_\mu X^\dagger_\nu X_\rho X^\dagger_\sigma
\end{equation}
where $t_{\mu\nu\rho\sigma}$ is a tensor which identifies indices pairwise, or two traces like
\begin{equation}
\sum_{\mu, \nu, \rho, \sigma} t_{\mu\nu\rho\sigma} \Tr (X_\mu X^\dagger_\nu)\ \Tr (X_\rho X^\dagger_\sigma)
\end{equation}
The connected, quartic interactions are thus (we identify the polynomial with its bubble)
\begin{align}
I_{p;1}(X, X^\dagger) = \sum_{\mu, \nu} \Tr \bigl(X_\mu X^\dagger_\mu X_\nu X^\dagger_\nu\bigr) &= \begin{array}{c}\includegraphics[scale=.25]{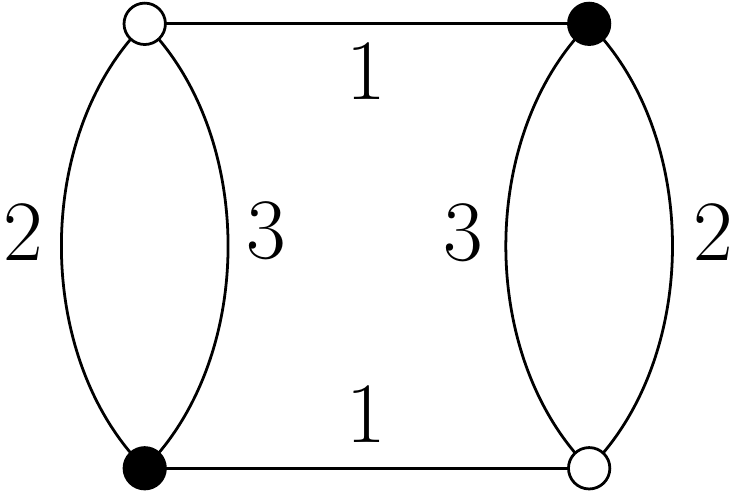}\end{array}\\
I_{p;2}(X, X^\dagger) = \sum_{\mu, \nu} \Tr \bigl(X_\mu X^\dagger_\nu X_\nu X^\dagger_\mu\bigr) &= \begin{array}{c}\includegraphics[scale=.25]{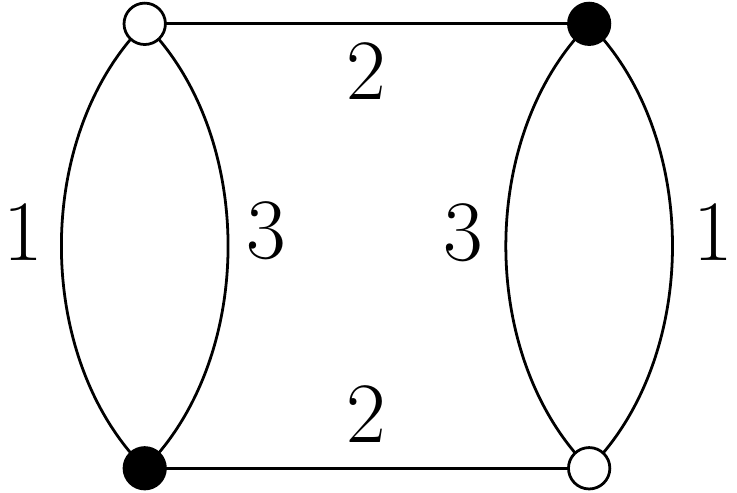}\end{array}\\
I_t(X, X^\dagger) = \sum_{\mu, \nu} \Tr \bigl(X_\mu X^\dagger_\nu X_\mu X^\dagger_\nu\bigr) &= \begin{array}{c}\includegraphics[scale=.25]{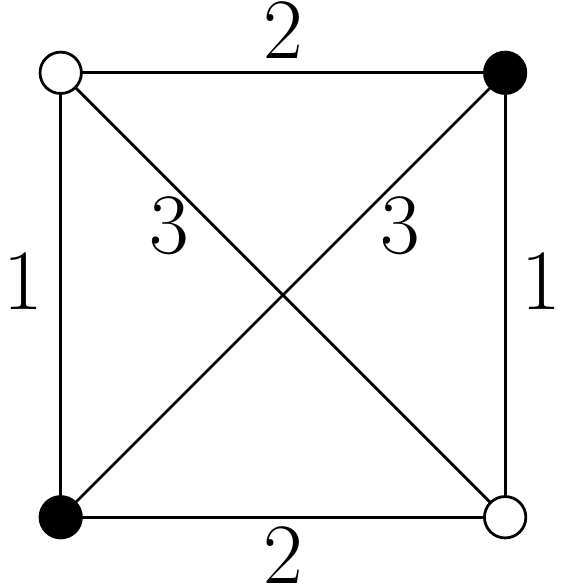}\end{array}\\
I_{p;3b}(X, X^\dagger) = \sum_{\mu, \nu} \Tr \bigl(X_\mu X^\dagger_\nu\bigr)\ \Tr \bigl(X_\nu X^\dagger_\mu\bigr) &= \begin{array}{c}\includegraphics[scale=.25]{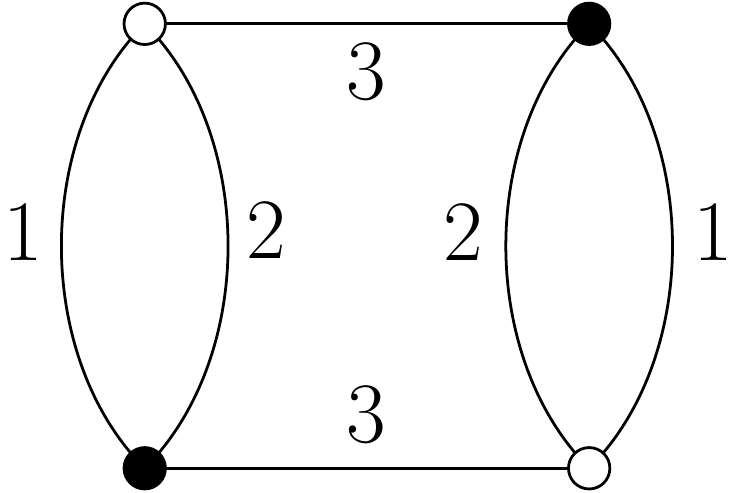}\end{array}\\
I_{p;3nb}(X, X^\dagger) = \sum_{\mu, \nu} \Tr \bigl(X_\mu X^\dagger_\nu\bigr)\ \Tr \bigl(X_\mu X^\dagger_\nu\bigr) &= \begin{array}{c}\includegraphics[scale=.25]{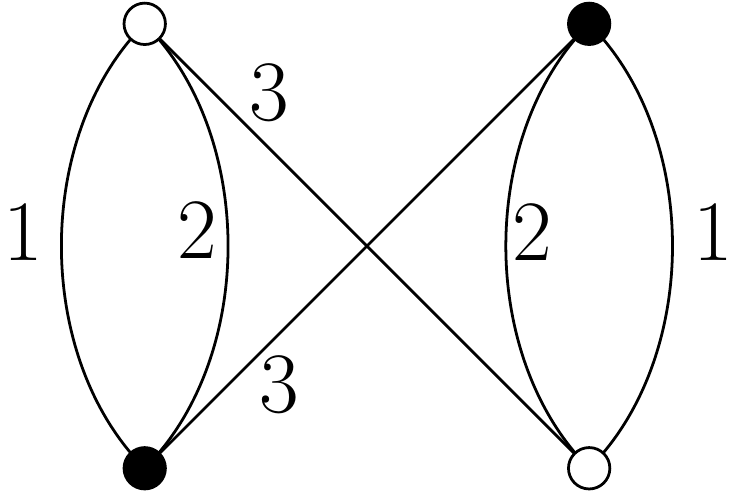}\end{array}
\end{align}

For the rest of this section, propagators will be represented as dashed edges between bubbles to distinguish them clearly from the bubbles' edges.

\paragraph{The model.\\}

The action of the model we study writes 
\begin{multline} \label{UN2ODAction}
S_{U(N)^2\times O(D)}(X_{\mu},X_{\mu}^{\dagger }) = -ND \sum_{\mu=1}^D \Tr(X_{\mu} X^\dagger_{\mu}) + N D^{3/2}\frac{\lambda_1}{2} I_t(X_{\mu},X_{\mu}^{\dagger }) \\
+ ND\frac{\lambda_2}{2} \bigl(I_{p;1}(X_{\mu},X_{\mu}^{\dagger })+I_{p;2}(X_{\mu},X_{\mu}^{\dagger })\bigr) +  D^2\frac{\lambda_2}{2} \bigl(I_{p;3b}(X_{\mu},X_{\mu}^{\dagger }) + I_{p;3nb}(X_{\mu},X_{\mu}^{\dagger })\bigr)
\end{multline}
One recovers the model studied in \cite{BeCa} by taking $\lambda_2=0$. Here we allow for all the connected, quartic interactions invariant under the the action of $U(N)^2\times O(D)$.
The scaling of the terms with respect to $N$ in the action is the ordinary one from multi-trace matrix model\footnote{In a multi-trace matrix model for a Hermitian matrix, the action is $\sum_{n>0} \sum_{\lambda_1\geq \lambda_2 \geq \dotsb \geq \lambda_n>0} N^{2-n} \prod_{j=1}^n \Tr M^{\lambda_j}$.}. The scaling behavior with respect to $D$ has been found in~\cite{FeVa} and is such that a non-trivial $1/D$ expansion exists. The partition function is
\begin{equation}
Z_{U(N)^2\times O(D)}(\lambda_1, \lambda_2) = e^{F_{U(N)^2\times O(D)}(\lambda_1, \lambda_2)} = \int \prod_{\mu=1}^D \prod_{a,b=1}^N d(X_\mu)_{ab} d(\overline{X_\mu})_{ab}\ e^{S_{N,D}(X_\mu, X^\dagger_\mu)}
\end{equation}

\paragraph{Vacuum graphs.\\}

Feynman graphs are obtained by taking any collection of interactions and performing Wick contractions. The latter can be represented graphically as edges pairing every white vertex to a black vertex. As usual in the literature, we give those edges the fictitious color 0. As a result the set of connected Feynman graphs, denoted $\bar{\mathbb{G}}$ is the set of connected, 4-regular, properly-edge-colored graphs such that the subgraph obtained by removing all edges of color 0 is a disjoint union of quartic bubbles $I_{p;1}, I_{p;2}, I_{p;3b}, I_{p;3nb}, I_t$.

Let $\bar{\cG}\in\bar{\mathbb{G}}$, then denote
\begin{itemize}
\item $n_t, n_{1}, n_{2}, n_{3b}, n_{3nb}$ the number of interactions $I_t, I_{p;1}, I_{p;2}, I_{p;3b}, I_{p;3nb}$ respectively,
\item $E_0$ the number of edges of color 0 (with $E_0 = 2(n_t+n_{p;1}+n_{p;2}+n_{p;3b}+n_{p;3nb})$),
\item $F_{0a}$ for $a=1,2,3$, the number of bicolored cycles which alternate the colors $0, a$. For $a=1,2$, we will also call them faces of colors 1 and 2 (the reason for this will be clear below).
\end{itemize}

The free energy has $1/N, 1/D$ expansions \cite{FeVa}. The $1/N$ expansion is governed by $h(\bar{\cG})$ and the $1/D$ expansion is governed by a non-negative integer, the \emph{grade} denoted $l(\bar{\cG})$.

\begin{theorem}\label{thm:free_energy}
The free energy expands as
\begin{equation}\label{F_exp}
F_{U(N)^2\times O(D)}(\lambda_1, \lambda_2) = \sum_{\bar{\cG}\in\bar{\mathbb{G}}} \biggl(\frac{N}{\sqrt{D}}\biggr)^{2-2h(\bar{\cG})} D^{2-l(\bar{\cG})/2}\ \mathcal{A}_{\bar{\cG}}(\lambda_1, \lambda_2)
\end{equation}
where $h(\bar{\cG})$ is a non-negative integer which we call the genus and $l(\bar{\cG})$ is a non-negative integer called the \emph{grade}.
\end{theorem}
Notice that in fact, the large $N$, large $D$ expansion is an expansion in $D$ and $L := \frac{N}{\sqrt{D}}$.

\begin{proof}
Here we give an alternative proof to \cite{FeVa}. From the Feynman expansion, one has
\begin{align} \label{GenusExpansion}
2-2h(\bar{\cG}) &= F_{01} + F_{02} - E_0 + n_t + n_{1} + n_{2}\\
1+h(\bar{\cG}) - \frac{l(\bar{\cG})}{2} &= F_{03} - E_0 +\frac{3}{2}n_t + n_{1} + n_{2} + 2(n_{3b}+n_{3nb}), \label{Grade}
\end{align}
and we want to show that both are non-negative.

Since we are primarily dealing with matrices, one expects the Feynman graphs to involve ribbon graphs, and the $1/N$ expansion to involve their genera, and thus be related to $h(\bar{\cG})$ and $l(\bar{\cG})$. Let us explain how ribbon graphs are encoded in our colored graphs. From $\bar{\cG}\in \bar{\mathbb{G}}$, remove all edges of color 3. Since they correspond to identifications of vector indices, it means that we are only keeping the information about the matrix part of the model. The resulting graph is a not-necessarily connected, 3-regular, properly-edge-colored graph, with colors 0, 1, 2.

Such graphs are known to be equivalent to ribbon graphs (more precisely of models for a complex matrix). One way to see this is to choose a cyclic order for the three colors meeting at every white vertex and the reverse cyclic order at the black vertices. For instance, one draws the edges of colors $(0,1,2)$ in the counter-clockwise order around the white vertices and in the clockwise order around black vertices. This provides each cycle of colors $\{1,2\}$ with a cyclic order of its incident edges of color 0. It can thus be replaced with a ribbon vertex,
\begin{equation}
\includegraphics[valign=c,scale=.5]{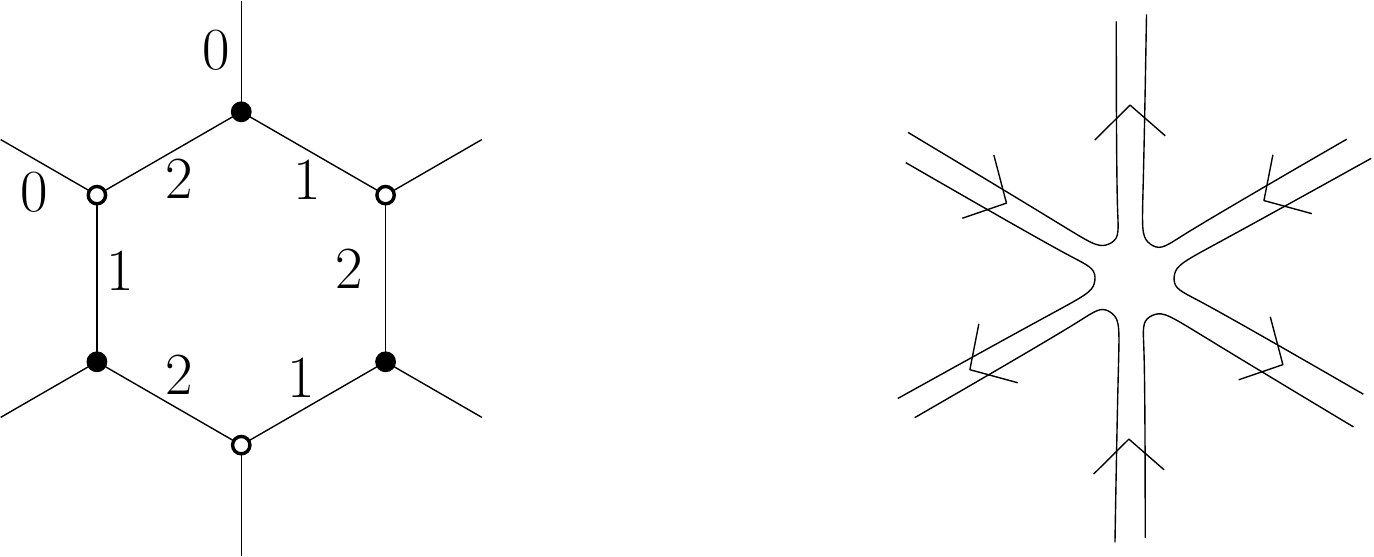}
\end{equation}
(of course in our case all cycles of colors $\{1,2\}$ are of length 4, but the correspondence is more general). Only edges of color 0 remain and they can be thickened (without twist) to obtain a ribbon graph denoted $\cG_{012}$. Arrows can be used to recover the vertex coloring of $\bar{\cG}$: for instance an edge of color 0 in $\bar{\cG}$ is oriented from black to white and those orientations are inherited in $\cG_{012}$. 

An alternative representation is as a combinatorial map, where one collapses the ribbon but keep the cyclic ordering of the edges around each vertex. This is enough to reproduce the ribbon graph. An example of a graph $\bar{\cG} \in \mathbb{G}_{U(N)^2\times O(D)}$ and its associated combinatorial map is given on Figure~\ref{fig:un2xod_to_map}. All vertices have degree four (because all bubbles are quartic).
\begin{figure}
\centering
\includegraphics[scale=0.7]{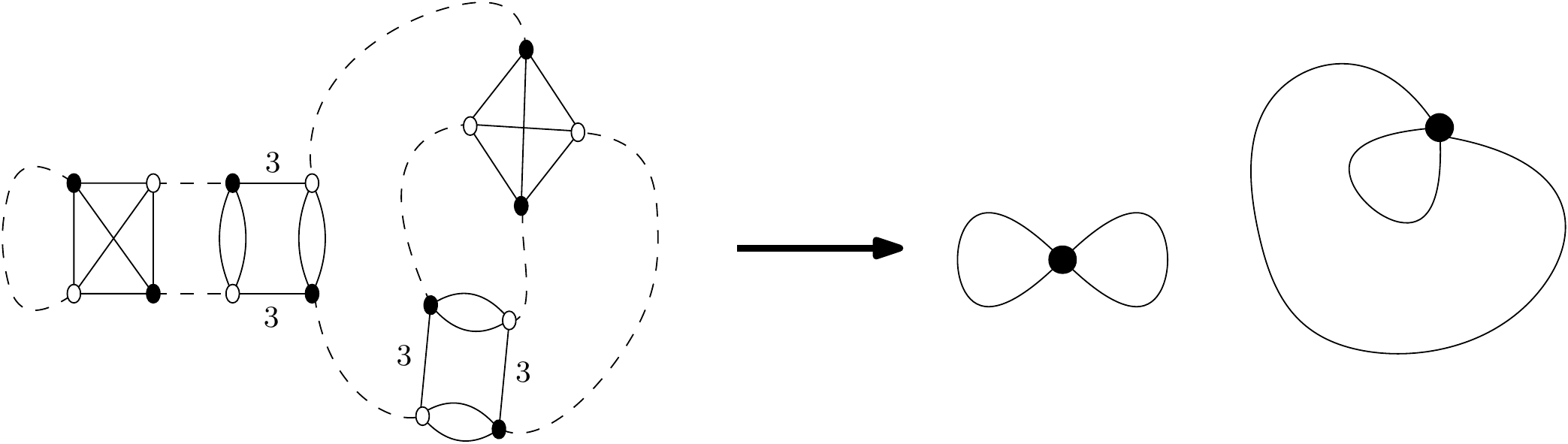}
\caption{A graph of the $U(N)^2\times O(D)$ model and its ribbon graph.}
\label{fig:un2xod_to_map}
\end{figure}

In addition to vertices and edges, a ribbon graph has faces, obtained by following the border of the ribbon. It is well-known that Euler's relation applies in this context: if $F, E, V$ are respectively the number of faces, edges and vertices of a connected ribbon graph, then
\begin{equation}
F-E+V=2-2h
\end{equation}
where $h$ is a non-negative integer called the genus of the ribbon graph.

In order to apply this to $\cG_{012}$, we first need to identify the faces and vertices, and to discuss connectedness. The faces of $\cG_{012}$ are the faces of colors 1 and 2 of $\bar{\cG}$, as seen from the bijection described above. Moreover,
\begin{equation} 
V(\cG_{012}) = n_t+n_{1}+n_{2}+2(n_{3b}+n_{3nb}).
\end{equation}
Note that $\cG_{012}$ may not be connected if $G$ contains some interactions $I_{p;3b}$ and $I_{p;3nb}$ because the latter are disconnected by the removal of their edges of color 3 (corresponding to the fact that they are double-trace interactions). The genus and number of connected components of this ribbon graph are denoted $g_{012}$ and $c_{012}$. Here the genus is the sum of the genera of its connected components. Euler's formula thus gives
\begin{equation}
2c_{012} - 2g_{012} = F_{01} +F_{02} - E_0 + n_t + n_{1} + n_{2} + 2(n_{3b}+n_{3nb}).
\end{equation}

As a result, Equation~\eqref{GenusExpansion} for $h(\bar{\cG})$ can be expressed in terms of the genus of the ribbon graph $\cG_{012}$ as
\begin{equation} \label{GenusExpansion2}
h(\bar{\cG}) = g_{012} + n_{3b}+n_{3nb}+1-c_{012}.
\end{equation}
Crucially, $g_{012}\geq 0$ and $n_{3b}+n_{3nb}+1-c_{012}\geq 0$ which is easily proved by induction for example, so that $h(\cG)\geq 0$. 
In the tetrahedral model, $n_{3b}=n_{3nb}=0$, then $c_{012}=1$ and $h(\bar{\mathcal{G}})$ reduces to the genus of the ribbon graph $g_{012}$ (and more generally, whenever only single trace interaction are considered, as often in the literature \cite{BeCa}). Although it is not the genus of a ribbon graph anymore in the presence of the pillow interactions, we will still call $h(\bar{\mathcal{G}})$ the \emph{genus} of $\bar{\mathcal{G}}$ for simplicity.

Another ribbon graph can be obtained by following the same procedure as the one leading to $\cG_{012}$, but starting with removing the edges of color 1 instead of color 3. The corresponding ribbon graph is denoted $\cG_{023}$. Similarly one obtains $\cG_{013}$ when one starts by removing the edges of color 2. The only difference with the procedure leading to $\cG_{012}$ is that now edge-twists have to be allowed and the genera may be half-integers. Euler's formulas are
\begin{equation}
\begin{aligned}
2c_{013} - 2g_{013} &= F_{01} + F_{03} - E_0 + n_t + n_1 + 2n_2 + n_{3b} + n_{3nb}\\
2c_{023} - 2g_{023} &= F_{02} + F_{03} - E_0 + n_t + 2n_1 + n_2 + n_{3b} + n_{3nb}.
\end{aligned}
\end{equation}
Equation \eqref{Grade} for $l(\bar{\cG})$ can thus be written as
\begin{equation}
\frac{l(\bar{\cG})}{2} = g_{023} + g_{013} + (n_{1}+1-c_{023}) + (n_{2}+1-c_{013}).
\end{equation}
All the quantities into brackets are non-negative integers, and the genera may be half-integers.
\end{proof}

For reference, Equations~\eqref{GenusExpansion} and~\eqref{Grade} give the following combinatorial expression of the grade
\begin{equation}
2-\frac{l(\bar{\cG})}{2} = \frac{F_1+F_2}{2} + F_3 - \frac{3}{2}E_0 + 2n_t + 2(n_{3b}+n_{3nb}) + \frac{3}{2}(n_1+n_2).
\end{equation}

\paragraph{2-point graphs.\\}

Due to the $U(N)^2\times O(D)$ invariance, the 2-point function is
\begin{equation}
\langle (X_\mu)_{ab} (\overline{X}_\nu)_{cd}\rangle = \frac{1}{N^2 D} G_{N,D}(\lambda_1, \lambda_2) \delta_{\mu, \nu} \delta_{ac} \delta_{bd},
\end{equation}
with $G_{N, D}(\lambda_1, \lambda_2) = \langle \sum_{\mu=1}^D \Tr X_\mu X^\dagger_\mu\rangle$. It has an expansion on 2-point graphs, whose set is denoted $\mathbb{G}$. A 2-point graph is like a vacuum graph with one white and one black vertex having exactly one half-edge of color 0 incident on them. 
From a graph $\cG\in\mathbb{G}$ one can obtain a unique vacuum graph $\bar{\cG}\in \bar{\mathbb{G}}$ called its \emph{closure}, by connecting the two half-edges of color 0. This vacuum graph is furthermore equipped with a marked edge (that obtained by connecting the two half-edges) called a \emph{root}. The set of rooted graphs is the set of Feynman graphs for the expansion of $G_{N,D}(\lambda_1, \lambda_2)$.
The amplitudes of $\cG$ and its closure $\bar{\cG}$ are the same, up to a factor $N^2 D$, since the difference between those graphs is one face of color 1, one face of color 2 (both contributing to a power of $N$), and one bicolored cycle of colors $\{0,3\}$ (contributing to $D$). We therefore simply define the genus and the grade of a 2-point graph to be those of its closure, i.e. $h(\cG):= h(\bar{\cG})$ and $l(\cG) := l(\bar{\cG})$. In combinatorics, it is common to work with rooted objects to make their enumeration easier as the presence of the root suppresses symmetry factors of the graphs. 


\subsubsection{Melons, dipoles and chains}

Here we review some structures which play a key role in our analysis. They are classical structures in the analysis of graphs coming from tensor models which can be adapted to multi-matrix models~\cite{BeCa}. We also refer to~\cite{Bonzom4} for details. The graphs which dominate at large $N$ and large $D$ (or in fact large $L$) are those of vanishing genus and grade. They turn out to be exactly the same graphs that dominate the large $N$ expansion of tensor models, called \emph{melonic graphs} or \emph{melons}.

\paragraph{Melons\\}

The elementary melons of our model are the following 2-point graphs:
\begin{equation}
\includegraphics[scale=.5, valign=c]{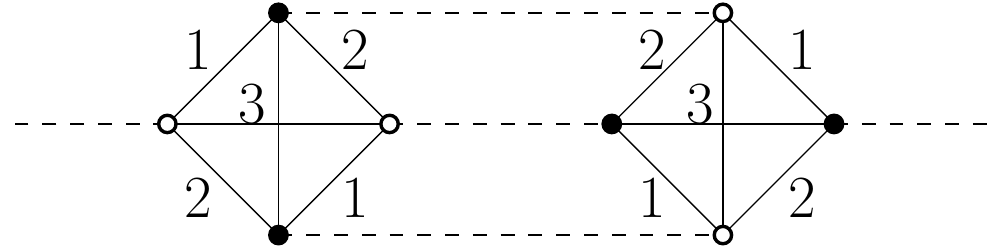}, \quad \includegraphics[scale=.5, valign=c]{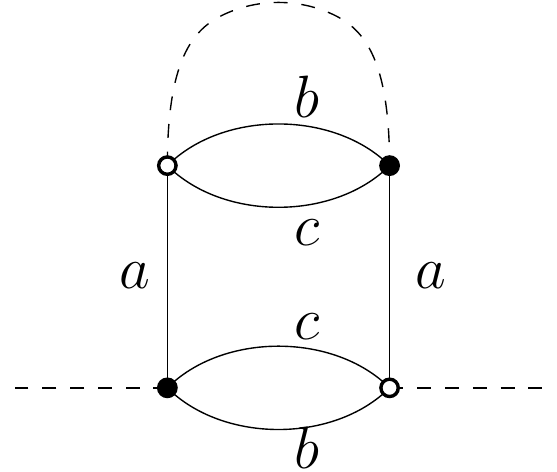}, \quad \includegraphics[scale=.5, valign=c]{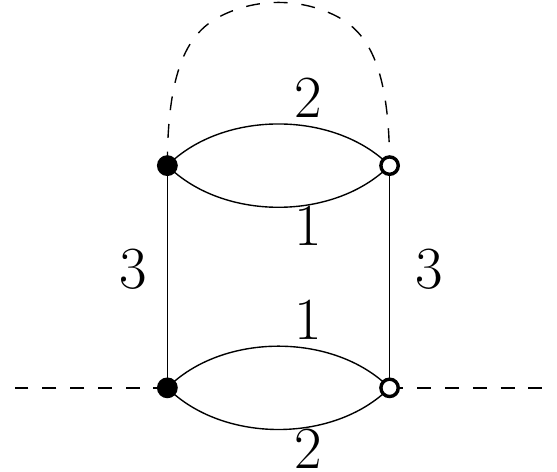}
\end{equation}
where $(a,b,c)$ is a permutation of $(1,2,3)$. Melonic graphs are obtained iteratively. First consider the closure of the elementary melons. Then on any edge of color 0, cut and insert any one of the elementary melons (still so that edges of color 0 connect white to black vertices), and repeat.

\begin{prop}
Inserting a melon on an edge of a graph $\bar{\cG}$ leaves $h(\bar{\cG})$ and $l(\bar{\cG})$ invariant. Melonic graphs are the only graphs of vanishing genus and grade. This is also true for 2-point graphs.
\end{prop}

From their recursive structure, it can be seen that the generating function of melonic 2-point graphs thus satisfies the equation
\begin{equation}
M(\lambda_1,\lambda_2) = 1 + \lambda_1^2M(\lambda_1,\lambda_2)^4 + 4\lambda_2M(\lambda_1,\lambda_2)^2
\label{eq:mel_un2}
\end{equation}
Performing the change of variable $(t,\mu) = \left(\lambda_1^2,\frac{4\lambda_2}{\lambda_1^2}\right)$ gives the same function $M(t,\mu)$ as in the $O(N)^3$-invariant tensor model. The behavior of this function has been studied in~\cite{TaCa}. In particular, its critical locus and its behaviour around this locus are known.

\vspace{5pt}
\paragraph{Dipoles\\}

Dipoles are the $4$-point graphs which can be obtained from an elementary melon by cutting one of its edges of color 0. We will group several dipoles together. The four half-edges naturally form two pairs: the pair from the elementary melon one started with and the pair from the edge which was cut and each pair will be called a \emph{side} of the dipole. We label the groups according to the color which is transmitted from one side of the dipole to the other. The dipoles of colors 1 and 2 are
\begin{equation}
D_{1,2} = \left\{\includegraphics[scale=0.22,valign=c]{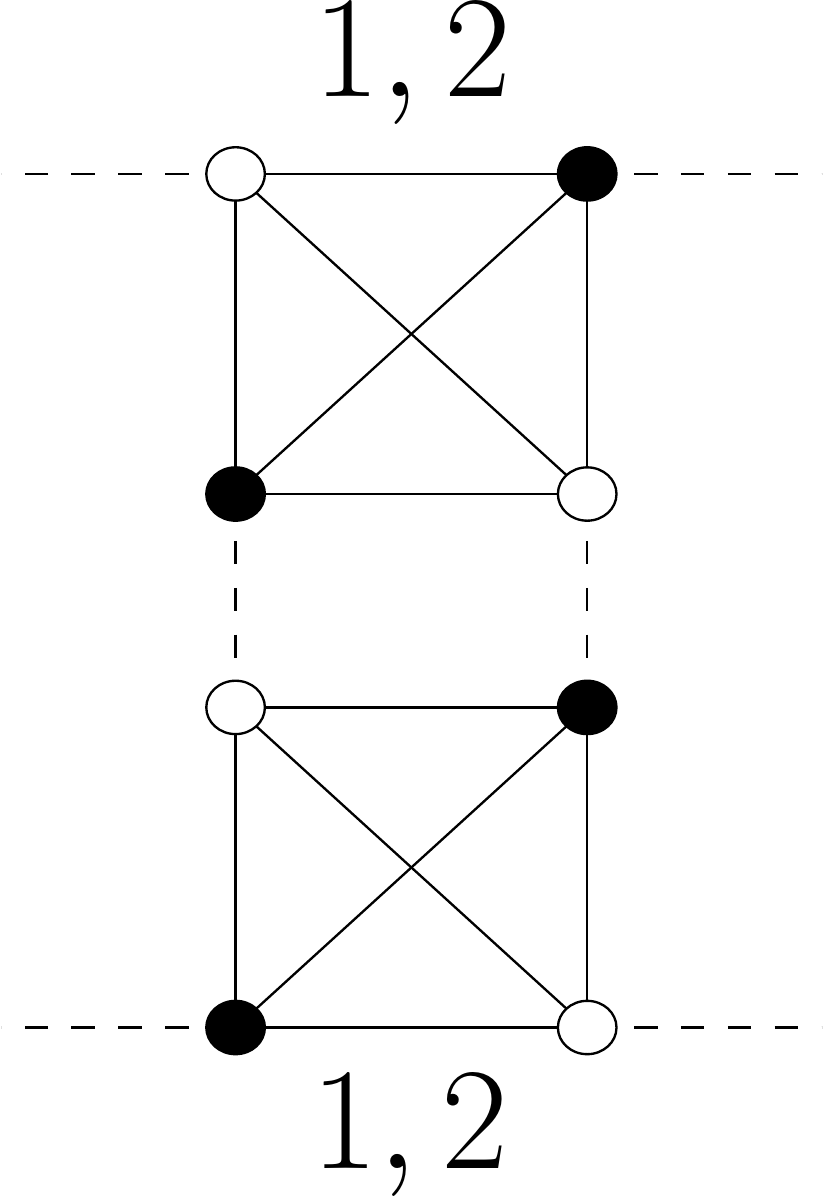},\hspace{.5cm} \includegraphics[scale=0.35,valign=c]{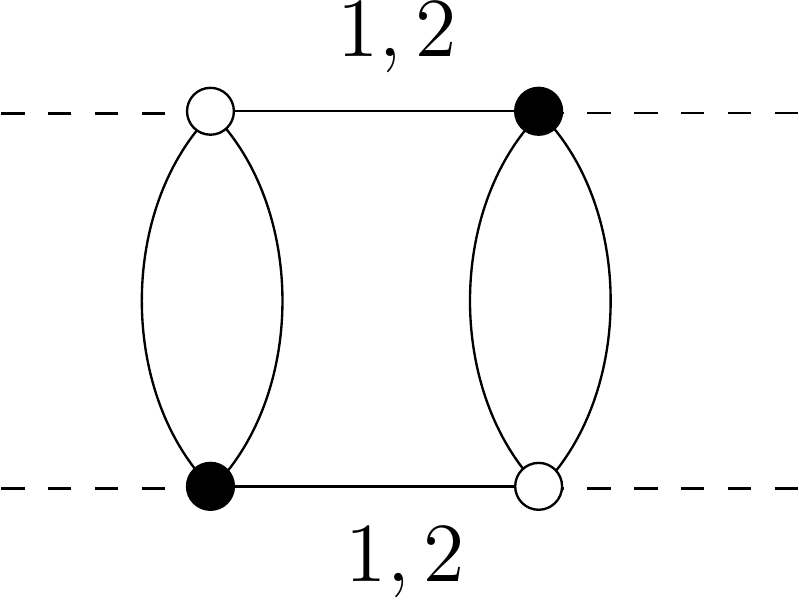}\right\}
\label{eq:d12}
\end{equation}
The first type is the L-, R-dipole of \cite{BeCa} (corresponding to the colors 1 and 2).

We further distinguish two types of dipoles of color $3$ according to the coloring of their vertices.
\begin{itemize}
\item One involving the tetrahedral interaction, and the non-bipartite pillow of color $3$, 
\begin{equation}
D_3 = \left\{\includegraphics[scale=0.25,valign=c]{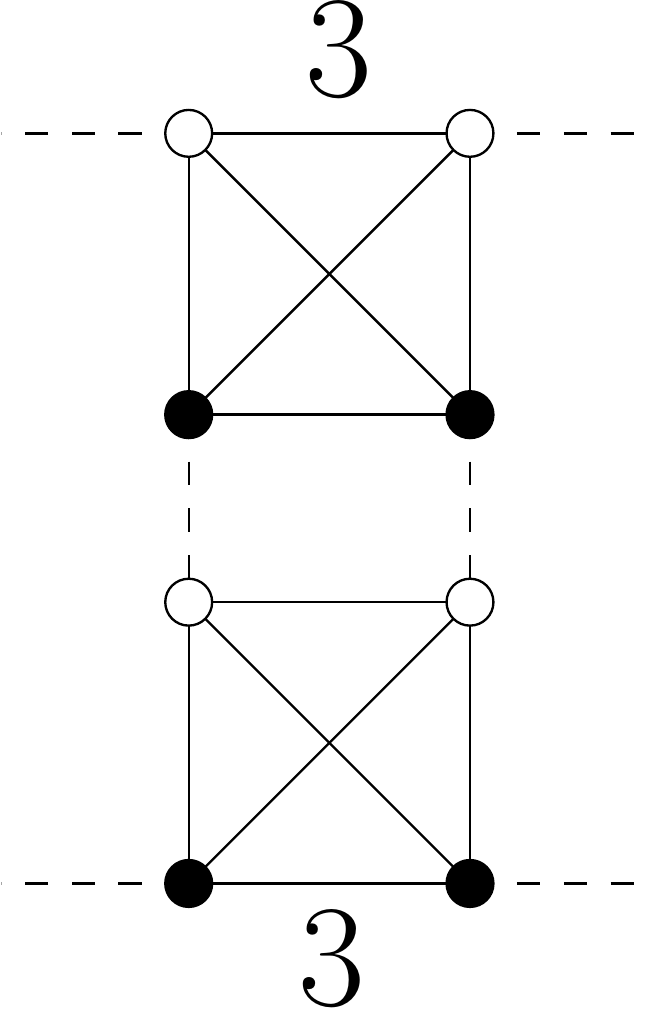}, \hspace{.5cm} \includegraphics[scale=0.26,valign=c]{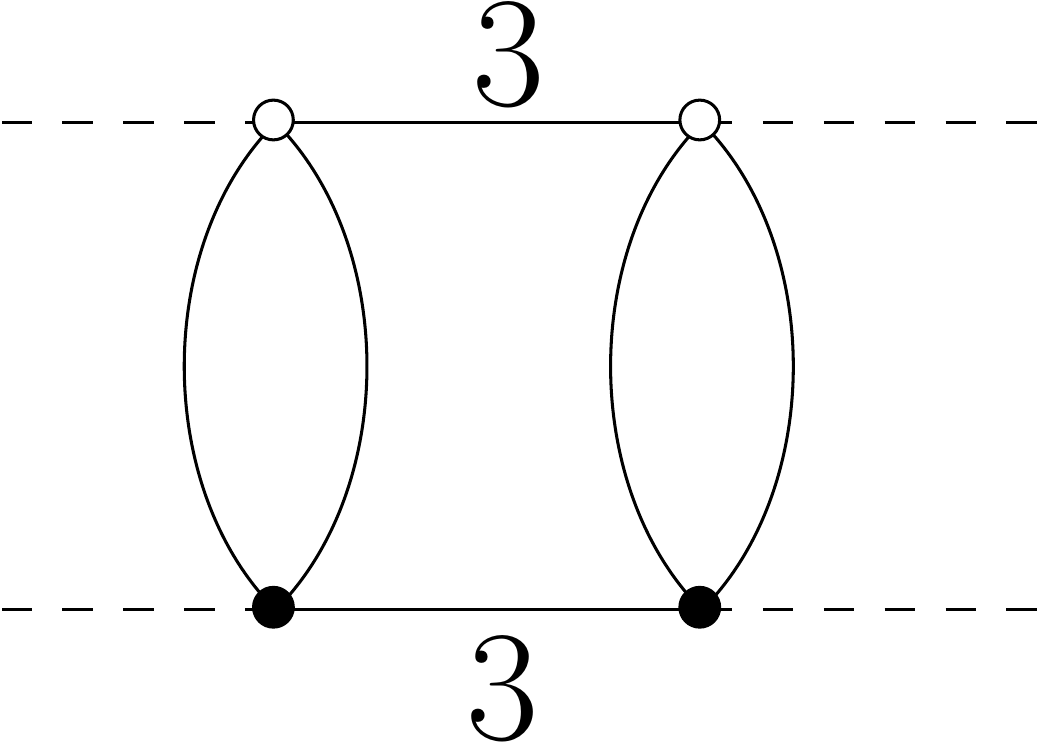}\right\}
\label{eq:d3}
\end{equation}
The first term one corresponds to the N-dipole of \cite{BeCa}.
\item One for the bipartite pillow of color $3$, 
\begin{equation}
D_3' = \left\{\includegraphics[scale=0.3,valign=c]{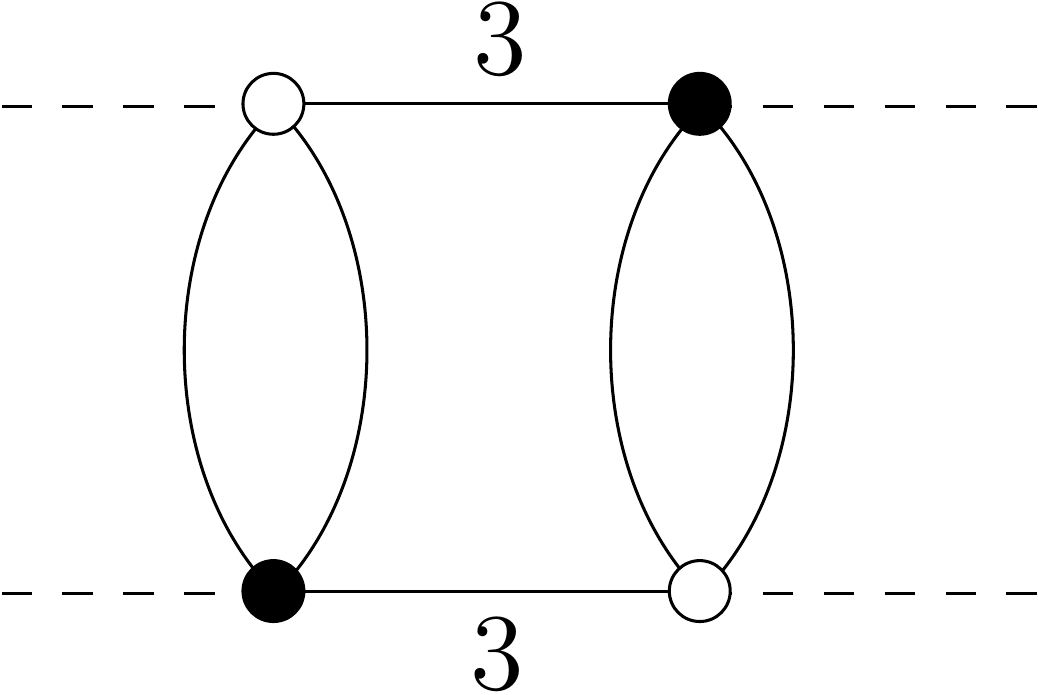}\right\}
\label{eq:d3'}
\end{equation}
\end{itemize}
All dipoles formed by pillows are obviously new compared to \cite{BeCa}.

Notice that dipoles may not be bubble-disjoint, i.e. two dipoles may share a bubble. We then say that the dipole is \emph{non-isolated}. It is easy to show that all non-isolated dipoles form a subgraph like in Figure \ref{fig:BubbleJointDipoles}. Dipoles which are bubble-disjoint from others are said to be \emph{isolated}.

\begin{figure}
    \centering
    \includegraphics[scale=.5]{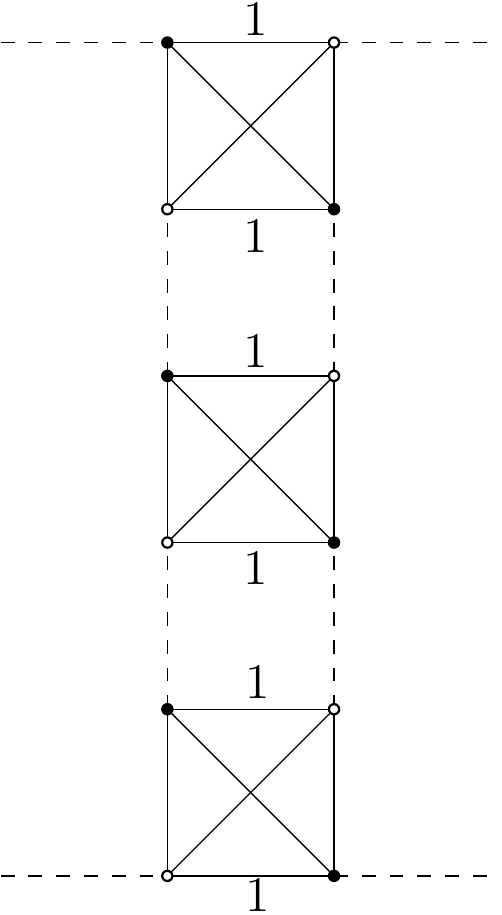}
    \caption{An example of two dipoles which are not bubble-disjoint.}
    \label{fig:BubbleJointDipoles}
\end{figure}

When they are isolated, it is useful to represent dipoles from the same group as a \emph{dipole-vertex}. In terms of generating series, the series associated to a dipole-vertex is the sum of the series of the dipoles in that group. We represent a dipole-vertex as a box with each pair on either sides. To distinguish the sides, we separate them with some thickened edges along the box, so that
\begin{equation}
    \begin{aligned}
    \includegraphics[scale=0.37,valign=c]{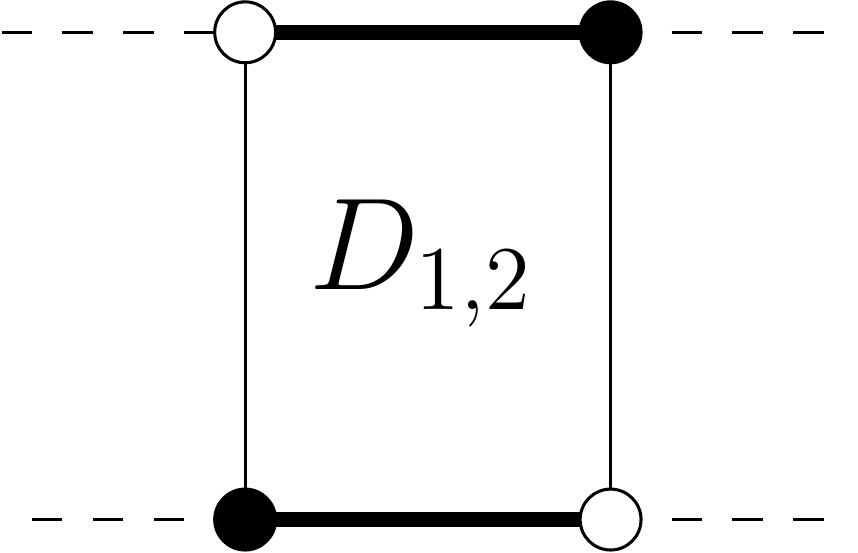} &= \includegraphics[scale=0.22,valign=c]{d12_m.pdf} + \includegraphics[scale=0.35,valign=c]{d12_pb.pdf}\\
    \includegraphics[scale=0.3,valign=c]{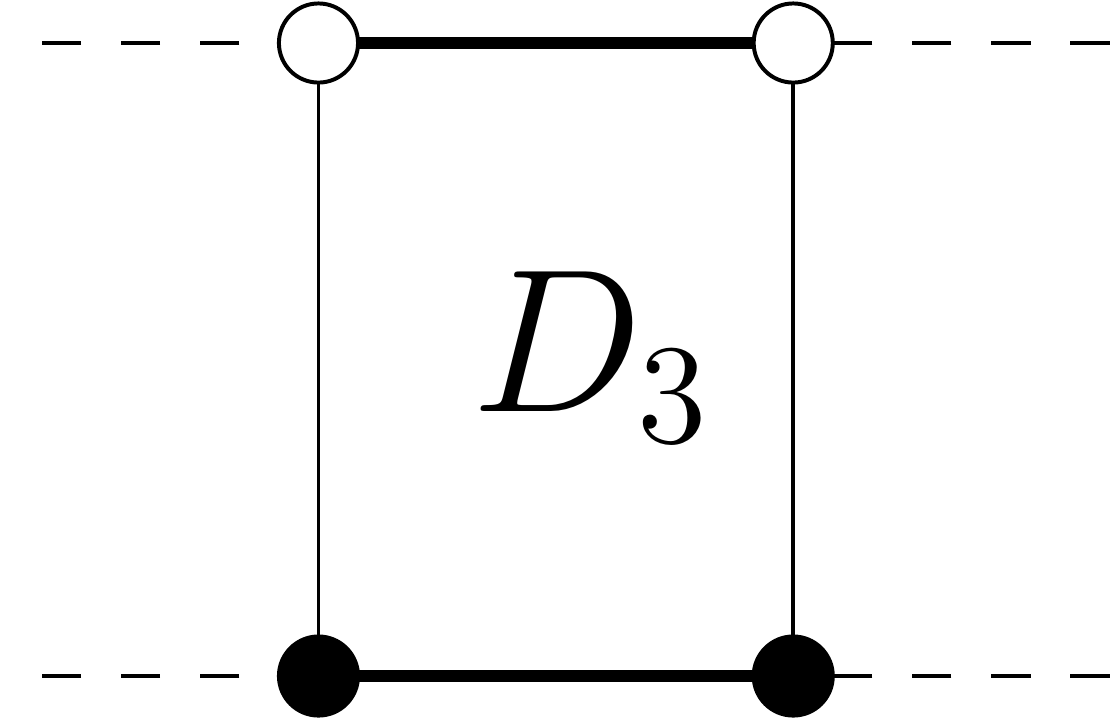} &= 
    \includegraphics[scale=0.25,valign=c]{d3_m.pdf} +  \includegraphics[scale=0.26,valign=c]{d3_pnb.pdf}\\
    \includegraphics[scale=0.45,valign=c]{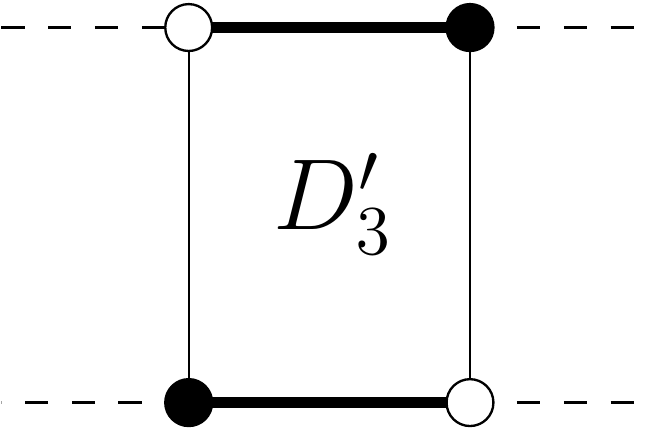} &= \includegraphics[scale=0.3,valign=c]{d3_pb.pdf}
    \end{aligned}
\end{equation}
Let us denote $U$ the generating function for the dipoles $D_{1,2,3}$ decorated with melons on one side and on the internal edges, and $V$ the generating function for $D'_3$ decorated with melons on one side. Then
\begin{align}
U(t,\mu) &= tM(t,\mu)^4 + \frac{1}{4}t\mu M(t,\mu)^2 \underset{\eqref{eq:mel_un2}}{=} M(t,\mu)-1-\frac{3}{4}t\mu M(t,\mu)^2 \\
V(t,\mu) &= \frac{1}{4}t\mu M(t,\mu)^2 
\end{align}

\paragraph{Chains\\}

A chain is either an isolated dipole, or a $4$-point function obtained by connecting dipoles side by side. The \emph{length} of a chain is the number of dipoles of the chain. Notice that a chain of length $\ell$ contains subchains of all lengths $1\leq \ell'\leq\ell$. A chain is said to be \emph{maximal} in a graph $\bar{\cG}$ if it cannot be included in a longer chain in $\bar{\cG}$. Two different maximal chains are necessarily \emph{bubble-disjoint}. Note that this would not be true if non-isolated dipoles were considered as chains, see Figure \ref{fig:BubbleJointDipoles}. The usual solution in the literature is to disallow chains from having a single dipole. Here, it is however allowed provided the dipole is bubble-disjoint from other dipoles\footnote{This will give us a slightly stronger result than what is usually found in the literature, since usually schemes can have dipoles (not included in chains) while here they can only have non-isolated dipoles.}.

A maximal chain will be represented by a chain-vertex, which is drawn exactly like a dipole-vertex, but labeled with a $C$ for chain,
\begin{equation} \label{Chain}	
\includegraphics[scale=.45, valign=c]{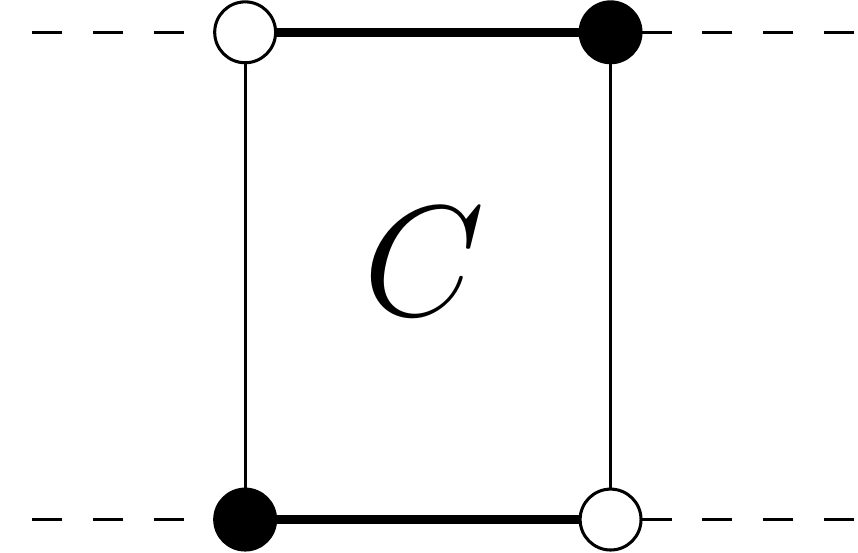} = \sum\limits_{k \geq 1} \includegraphics[scale=.50, valign=c]{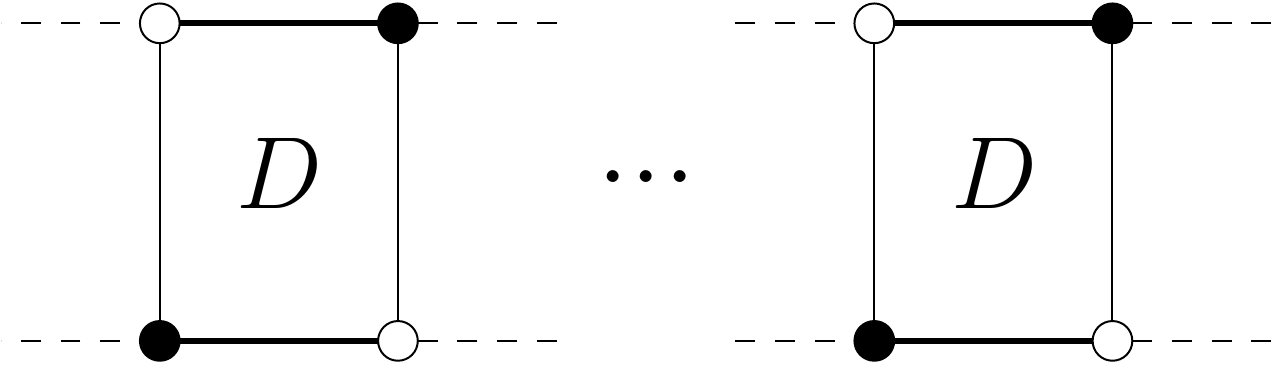}
\end{equation}
A chain of color $c$ is a chain which involves only dipoles of color $c$. If not, it is said to be \textit{broken}. Chain-vertices can thus be labeled $C_c$ if they represent a chain of color $c$, or labeled $B$ if they represent a broken chain. Note that the minimal realization of a chain-vertex of color $c$ is as a single dipole of color $c$, while the minimal realization of a broken chain-vertex is as a chain with 2 dipoles of different colors.

The following proposition is easy to prove.
\begin{prop} \label{thm:ChainLength}
Changing the length of a chain does not change the genus nor the grade (both in vacuum and 2-point graphs).
\end{prop}

In terms of generating series, chains of color 1 and 2 are geometric series of dipoles of color 1 and 2. Chains of color 3 can have at each step either $D_3$ or $D'_3$. The corresponding generating functions are
\begin{equation}
C_1 = C_2 = U \sum_{n\geq0} U^n = \frac{U}{1-U} \qquad
C_3 = (U+V) \sum\limits_{n\geq0} (U+V)^n = \frac{(U+V)}{1-(U+V)}
\label{eq:chain_un2}
\end{equation}
We can further distinguish two types of chains of color $3$, depending on the coloring on the vertex at their ends. We denote the respective generating series $C_{3,\includegraphics[scale=0.2,valign=c]{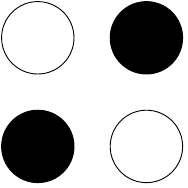}}$ and $C_{3,\includegraphics[scale=0.2,valign=c]{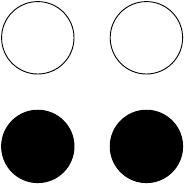}}$. To evaluate them, notice that $C_{3,\includegraphics[scale=0.2,valign=c]{wbbw.pdf}}$ (respectively $C_{3,\includegraphics[scale=0.2,valign=c]{wwbb.pdf}}$) is a series in $U$ and $V$ where each term has an even (respectively odd) number of $U$. Thus, we introduce a dummy variable $x$ which counts the number of $U$ in each term. We have
\begin{align}
C_{3,\includegraphics[scale=0.2,valign=c]{wbbw.pdf}} = \sum_{k\geq 0} [x^{2k}] \frac{xU+V}{1-(xU+V)} = \frac{U^2+V-V^2}{(1-V)^2-U^2}\\
C_{3,\includegraphics[scale=0.2,valign=c]{wwbb.pdf}} = \sum_{k\geq 0} [x^{2k+1}] \frac{xU+V}{1-(xU+V)}= \frac{U}{(1-V)^2-U^2}
\end{align}

Finally the generating function of broken chains is that of all remaining chains.
\begin{equation} \label{BrokenChainsUN2OD}
\begin{aligned}
B &= \left( 3U+V \right) \sum_{n\geq 0} \left( 3U+V \right)^n - \sum\limits_{i=1}^{3} C_i = \frac{\left(3U+V\right)}{1-3U-V} - 2\frac{U}{1-U} - \frac{U+V}{1-U-V} \\
&= \frac{-6U^3 -8U^2V+6U^2-2UV^2 +4 UV}{\left(1-3U-V\right)\left(1-U\right)\left(1-U-V\right)}
\end{aligned}
\end{equation}
Similarly as for chains of color $3$, we could distinguish two types of broken chains depending on the coloring of the vertices at its boundary. However, that will not be necessary for our analysis.

\subsubsection{Schemes}

Let $\cG\in\bar{\mathbb{G}}$. The scheme of $\cG$ is obtained by 
\begin{enumerate}
\item Replacing every melonic 2-point subgraph with an edge of color 0,
\item Replacing every maximal chain with a chain-vertex of the same type,
\end{enumerate}
All graphs which reduce to the same scheme have the same genus $h$ and grade $l$, so it makes sense to define the genus and grade of a scheme as those values. We denote $\mathbb{S}_{g,l}$ the set of schemes at fixed values of $g, l$. We will prove the following result

\begin{theorem} \label{thm:FiniteSchemesI}
$\mathbb{S}_{g,l}$ is a finite set.
\end{theorem}

In other words, there is a finite number of schemes of fixed genus and grade. All graphs of fixed genus and grade $(g,l)$ can be obtained starting with a scheme in $\mathbb{S}_{g,l}$ and replacing 
chain-vertices with chains and edges of color 0 with melons.

\subsection{Finiteness of the number of schemes of genus \texorpdfstring{$g$}{g} and grade \texorpdfstring{$l$}{l}}

\subsubsection{Brief review of the \texorpdfstring{$O(N)^3$}{O(N)3} tensor model}

To show that there are finitely many schemes of genus $g$ and grade $l$, we are going to rely on the fact that a similar result holds for the quartic $O(N)^3$-invariant tensor model as shown in~\cite{Bonzom4}. The $O(N)^3$-invariant tensor involves a real tensor $\phi_{abc}$, with each index $a,b,c$ ranging in $\{1, \dotsc, N\}$. Interactions are required to have an $O(N)^3$-invariance, where each copy of the $O(N)$ group acts separately on an index of the tensor:
\begin{equation}
\phi_{abc} \rightarrow \phi'_{a'b'c'} = \sum_{a,b,c=1}^N O_{a'a}^1 O_{b'b}^2 O_{c'c}^3 \phi_{abc} \qquad O^i \in O(N)
\end{equation}
This model has four different quartic interactions, a tetrahedral one and a pillow one for each color. By drawing a vertex for each $\phi$ and an edge of color $c\in\{1,2,3\}$ for each index in position $c$ contracted between two $\phi$s, one obtains the following representation,
\begin{align}
I_t(\phi) &= \sum_{a, a', b, b', c, c'} \phi_{abc}\phi_{ab'c'}\phi_{a'bc'}\phi_{a'b'c} = \begin{array}{c}\includegraphics[scale=.25]{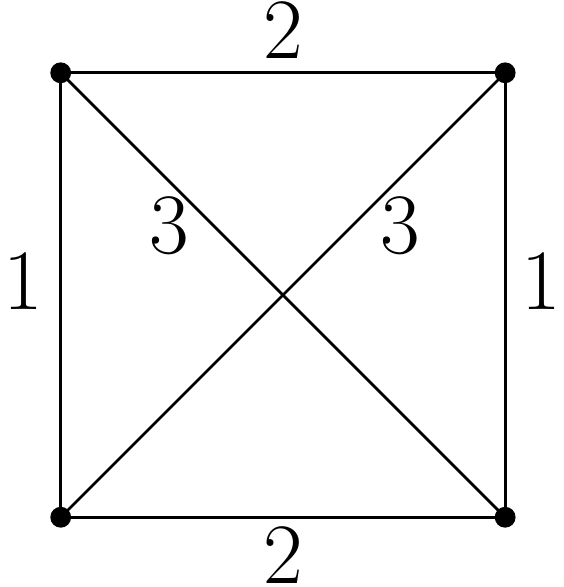}\end{array}\\
I_{p,1}(\phi) &= \sum_{a, a', b, b', c, c'} \phi_{abc}\phi_{a'bc}\ \phi_{ab'c'}\phi_{a'b'c'} = \begin{array}{c}\includegraphics[scale=.25]{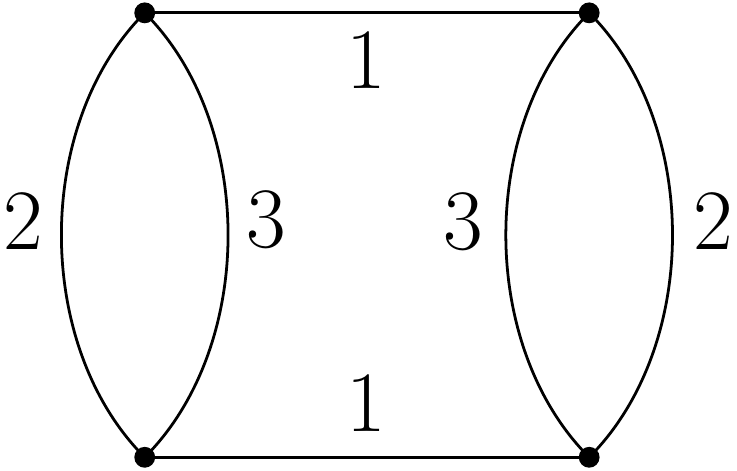}\end{array}\\
I_{p,2}(\phi) &= \sum_{a, a', b, b', c, c'} \phi_{abc}\phi_{ab'c}\ \phi_{a'bc'}\phi_{a'b'c'} = \begin{array}{c}\includegraphics[scale=.25]{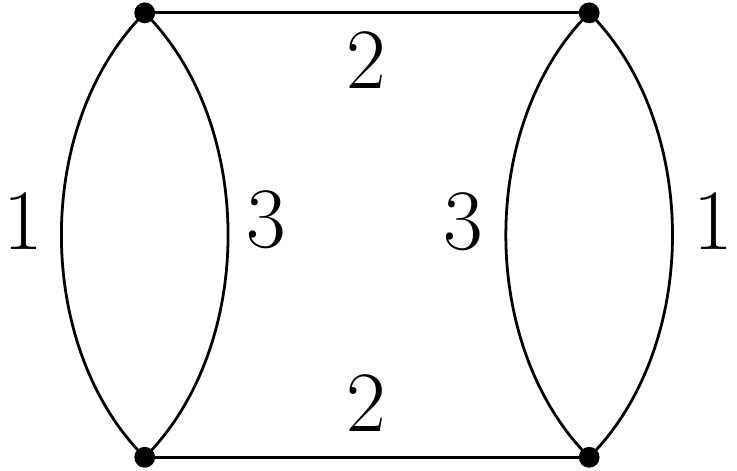}\end{array}\\
I_{p,3}(\phi) &= \sum_{a, a', b, b', c, c'} \phi_{abc}\phi_{abc'}\ \phi_{a'b'c}\phi_{a'b'c'} = \begin{array}{c}\includegraphics[scale=.25]{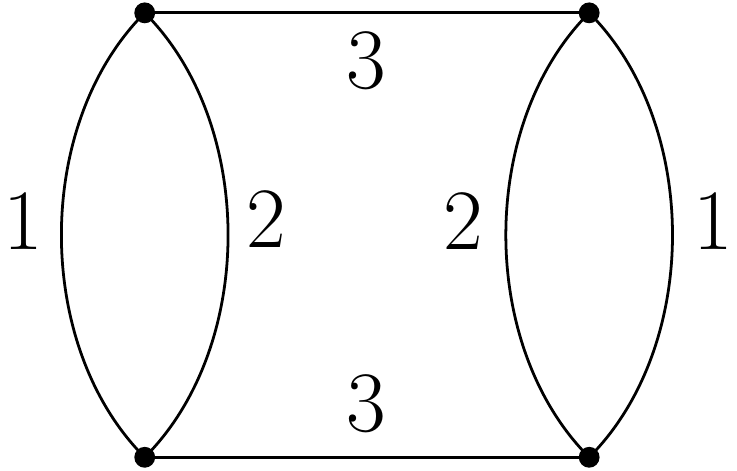}\end{array}
\end{align}
The action is
\begin{equation} \label{O(N)3Action}
S_N (\phi) = -\frac{N^{2}}{2}\sum_{a,b,c=1}^N \phi_{abc}^2 + N^{\frac{5}{2}} \frac{\lambda_1}{4}I_t(\phi) + N^2 \frac{\lambda_2}{4}\Bigl(I_{p,1}(\phi) + I_{p, 2}(\phi) + I_{p, 3}(\phi)\Bigr).
\end{equation}
The normalization here differs from \cite{Bonzom4} by a rescaling of $\phi$ with $N^{1/4}$. We use this convention so that setting $D=N$ in~\eqref{UN2ODAction} leads to the same scaling behavior as~\eqref{O(N)3Action}.

The Feynman graphs are obtained by taking a collection of interactions and connecting $\phi$'s pairwise via propagators, which we draw as edges of color 0. We denote $\bar{\mathbb{G}}_{O(N)^3}, \mathbb{G}_{O(N)^3}$ respectively the sets of (rooted) vacuum and 2-point graphs of the $O(N)^3$ model.

All bicolored cycles alternating colors $0$ and $c\in\{1, 2, 3\}$ contribute to a factor of $N$ in a Feynman amplitude (while only those with $c=1,2$ give a factor of $N$ in the multi-matrix model, since those with $c=3$ contribute with $D$ instead). The free energy has the expansion
\begin{equation}
F_{O(N)^3} = \sum_{\bar{\cG}\in\bar{\mathbb{G}}_{O(N)^3}} N^{3-\omega(\bar{\cG})} A(\bar{\cG}),
\end{equation}
where $A(\bar{\cG})$ only depends on the coupling constants, and $\omega(\bar{\cG})$ is a non-negative integer called the \emph{degree} and is given by
\begin{equation}
\omega (\bar{\cG}) = 3 + \frac{3}{2}n_t(\bar{\cG}) + 2n_p(\bar{\cG}) - \sum_{c=1}^3 F_c(\bar{\cG}),
\label{eq:deg}
\end{equation}
where
\begin{itemize}
\item $n_t(\bar{\cG})$ and $n_p(\bar{\cG})$ are respectively the number of tetrahedral and pillow bubbles in the graph,
\item $F_c(\bar{\cG})$ is the number of bicolored cycles with colors $\{0,c\}$.
\end{itemize}
For a 2-point graph $\cG\in\mathbb{G}_{O(N)^3}$, we set $\omega(\cG) = \omega(\bar{\cG})$ where $\bar{\cG}$ is its closure.

\subsubsection{Link between the graphs of two models}


Graphically, the only difference between the graphs of $\mathbb{G}$ and those of $\mathbb{G}_{O(N)^3}$ is the vertex coloring. There are two vertex colors in $\mathbb{G}$ because there are two fields $X, X^\dagger$, but no coloring in $\mathbb{G}_{O(N)^3}$ because $\phi$ is real. Hence there is a map 
\begin{equation}\label{eq:theta}
\theta: \mathbb{G} \to \mathbb{G}_{O(N)^3}
\end{equation}
which simply consists in forgetting the vertex coloring.

If $\mathcal{G}\in\mathbb{G}_{U(N)^2\times O(D)}$, we define the degree of $\cG$ as $\omega(\mathcal{G}) := \omega(\theta(\mathcal{G}))$ the degree of the graph it is mapped to in the $O(N)^3$ model. Then
\begin{equation}
\omega(\mathcal{G}) = h(\mathcal{G}) + \frac{l(\mathcal{G})}{2}
\end{equation}
which can be established by setting $D=N$ in \eqref{UN2ODAction}.

Melons, dipoles, chains and schemes are defined similarly in the $O(N)^3$-invariant model as here. We refer to~\cite{Bonzom4} for details. The map $\theta$ maps melons to melons, dipoles to dipoles and chains to chains. We denote $\mathbb{S}_{O(N)^3}(\omega)$ be the set of schemes of degree $\omega$ in the $O(N)^3$ model. We can therefore descend the map $\theta$ to schemes by
\begin{equation}
\tilde{\theta}_{h,l}: \mathbb{S}_{h,l} \to \mathbb{S}_{O(N)^3}(h+l/2)
\end{equation}
from the schemes of genus $h$ and grade $l$ to those of the $O(N^3)$ model with degree $h+l/2$. It simply consists in forgetting the coloring of the vertices on the vertices, dipoles and chains.

The main theorem of \cite{Bonzom4} is the following.
\begin{theorem}
$\mathbb{S}_{O(N)^3}(\omega)$ is finite.
\end{theorem}
It is therefore enough for us to now show that each scheme $\mathcal{S}\in \mathbb{S}_{O(N)^3}(h+l/2)$ has a finite fiber via $\tilde{\theta}_{h,l}$. This is clear since the fiber of $\mathcal{S}\in \mathbb{S}_{O(N)^3}(\omega)$ is found by considering all colorings of the vertices and chain-vertices. This proves Theorem \ref{thm:FiniteSchemesI}.

\begin{figure}[!ht]
\begin{center}
\includegraphics[scale=0.60]{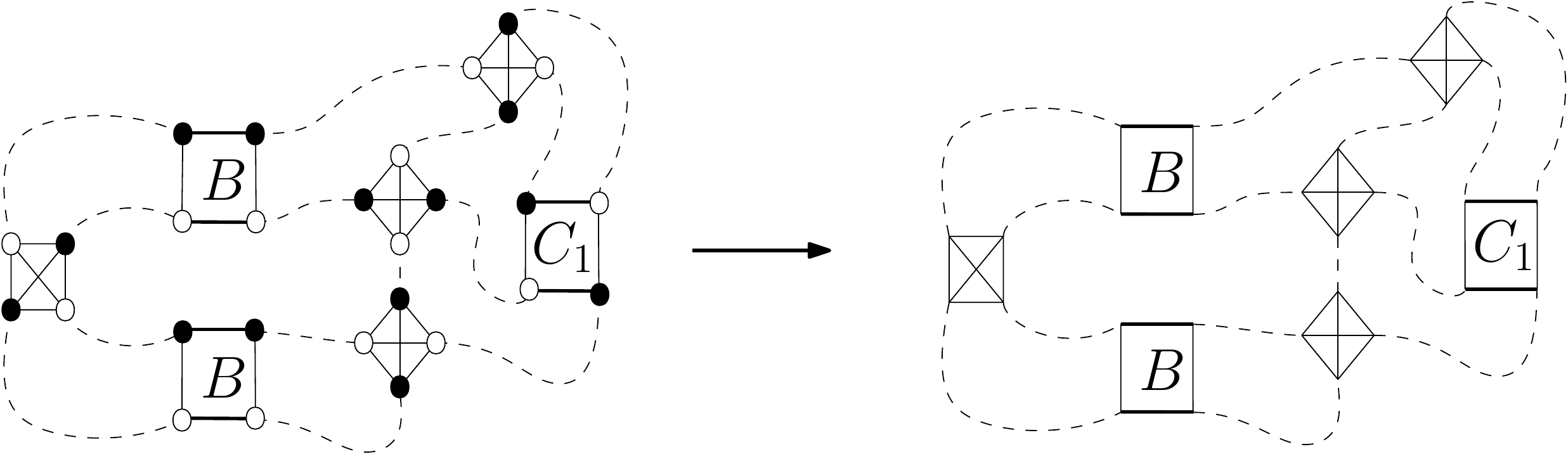}
\caption{A scheme of the $U(N)^2\times O(D)$-invariant model and its corresponding scheme on the side of the $O(N)^3$-invariant model.}
\end{center}
\end{figure}

For a fixed degree $\omega$, schemes of $\mathbb{S}_{O(N)^3}(\omega)$ have a bounded number of chains. Therefore, for a fixed value of $(g,l)$, the number of chains of schemes in $\mathbb{S}_{h,l}$ are also bounded. Note however that, as we will see later on, this bound is not the same for the two models as the mapping $\theta$ is not surjective.

\subsection{Identification of the dominant schemes and double scaling limit}

\subsubsection{Singularity analysis}  \label{sec:UN2ODSingularity}

Looking at the form of the generating series above, there are four different types of possible singular points:
\begin{itemize}
\item Singular points of $M(t,\mu)$,
\item Points where $U(t,\mu)=1$, these points are singular for chains of colors 1 and 2, and for broken chains,
\item Points where $1-U(t,\mu)-V(t,\mu)=0$, which are singular for chains of color 3 and for broken chains,
\item Points such that $1-V(t,\mu)+U(t,\mu)=0$, which are singular for even and odd chains of color 3,
\item Points such that $1-3U(t,\mu)-V(t,\mu)=0$, which are singular for broken chains.
\end{itemize}
Since $M(t, \mu)$ is an increasing function of $t$, so are $U(t,\mu)$ and $V(t,\mu)$, which are also non-negative since $M(0,\mu) = 1$. Therefore, at fixed $\mu$, the singular points $1-3U(t,\mu)-V(t,\mu)=0$ are always encountered before any other point, except for potential singular points of $M(t,\mu)$. Noticing that $3U(t,\mu)+V(t,\mu) = tM(t,\mu)^4 + t\mu M(t,\mu)^2$, the analysis performed in~\cite{Bonzom4} for the $O(N)^3$-invariant tensor model can be applied here. Thus, points where $3U(t,\mu)+V(t,\mu) = 1$ correspond to the locus of singular points of $M(t,\mu)$. These points are plotted on Figure~\ref{fig:plot_tc}.

\begin{figure}
\begin{center}
\includegraphics[scale=0.6]{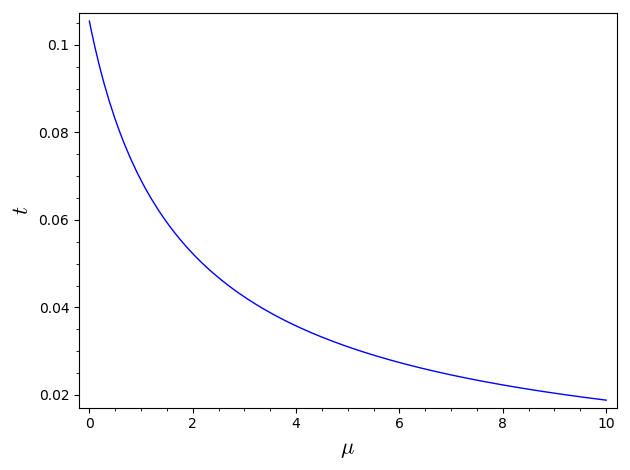}
\caption{Critical points for the generating function $M(t,\mu)$. They also correspond to points where $3U(t,\mu) + V(t,\mu) = \frac{1}{3}$ and therefore are critical points for $B(t,\mu)$ as well. This function is the same as the one governing the leading order of the quartic $O(N)^3$ model in~\cite{Bonzom4}.}
\label{fig:plot_tc}
\end{center}
\end{figure}

The behaviour of $M(t,\mu)$ near singular points has been studied in~\cite{TaCa}. For fixed $\mu$, there is a single critical value of $t$, denoted $t_c(\mu)$ where $M(t,\mu)$ is singular near the critical point $(\mu,t_c(\mu)$
\begin{equation}
M(t,\mu) \underset{t \rightarrow t_c(\mu)}{\sim} M_c(\mu) + K(\mu)\sqrt{1 - \frac{t}{t_c(\mu)}}
\label{eq:crit_behav}
\end{equation}
where $M_c(\mu)$ is the unique positive real root of the polynomial equation 
\begin{equation} \label{eq:Mc_poly}
-3x^3+4x^2-\mu x +2\mu = 0,
\end{equation}
and
\begin{equation}
K(\mu) = \sqrt{\frac{M_c(\mu)^2\left(M_c(\mu)^2 + \mu \right)}{6M_c(\mu)^2+\mu}}.
\end{equation}

Using Equation~\eqref{eq:crit_behav}, near a singular point $(t_c(\mu),\mu)$ of $M(t,\mu)$ we have:
\begin{align}
V(t,\mu) = \frac{1}{4}t_c(\mu)\mu \left( M_c(\mu) + 2M_c\mu)K(\mu)  \sqrt{1 - \frac{t}{t_c(\mu)}} \right) + \mathcal{O}( t_c(\mu) - t) \\
U(t,\mu) = M_c(\mu)-1 - \frac{3}{4} t_c(\mu) \mu M_c(\mu)^2 -\frac{3}{2} t_c(\mu) \mu M_c(\mu) K(\mu) \sqrt{1- \frac{t}{t_c(\mu)}} + \mathcal{O}( t_c(\mu) - t)
\end{align}
which leads to
\begin{align}
B(t,\mu) \underset{t\rightarrow t_c(\mu)}{\sim} &\left[\frac{-6U(t,\mu)^3 -8U(t,\mu)^2V(t,\mu)+6U(t,\mu)^2-2U(t,\mu)V(t,\mu)^2 +4 U(t,\mu)V(t,\mu)}{\left(1-U(t,\mu)\right)\left(1-U(t,\mu)-V(t,\mu)\right)}\right]_{\big\rvert_{t_c(\mu)}} \nonumber\\ &\times \frac{1}{\left(1-\frac{4}{3} t_c(\mu) \mu M_c(\mu) \right) \sqrt{1- \frac{t}{t_c(\mu)}}}
\end{align}

\subsubsection{Dominant schemes}

There are finitely many schemes for a given genus and grade $(h,l)$, thus all singularities of their generating function come from the generating series of melons and chains. As seen in Section~\ref{sec:UN2ODSingularity}, the leading singularity is that of broken chains, which we will show are in bounded number in a scheme of fixed genus and grade. Moreover, we restrict attention to schemes of vanishing grade as they are the ones dominating the large $N$, large $D$ expansion, just as in \cite{BeCa}. The schemes of vanishing grade which have the maximal number of broken chains for their genus are said to be \emph{dominant}.

Note that those dominant schemes cannot be found from the equivalent result for the $O(N)^3$ model. Indeed, since we have embedded our model into the $O(N)^3$ model and used that to prove the finiteness of the number of schemes at fixed genus and grade, it is a natural question to ask whether the latter model can also be used to identify the dominant schemes. However, dominant schemes of the $O(N)^3$ have sub-schemes of degree $\omega=1/2$ (the leaves from the tree in the representation explained below). A graph of the $U(N)^2\times O(D)$ model which is mapped to a graph of degree 1/2 of the $O(N)^3$ model has non-vanishing grade. It is thus not possible to deduce the dominant schemes of the $U(N)^2\times O(D)$ model from those of the $O(N)^3$ model.

The combinatorial analysis of dominant schemes for the $U(N)^2\times O(D)$-invariant model with tetrahedral interaction only has been done in~\cite{BeCa}. It can be adapted to our situation with minor adjustments. We sketch the main points of the derivation.

\subsubsection{Removals of chains and dipoles} \label{sec:DipoleRemovals}

The removal of a dipole or chain-vertex is the following move
\begin{equation}
    \begin{array}{c}\includegraphics[scale=.5]{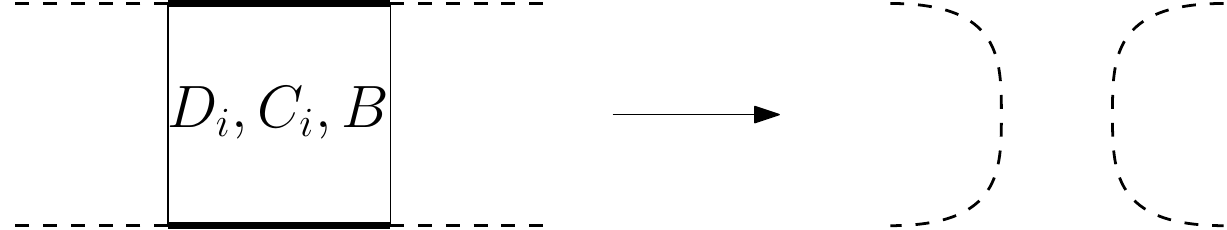}\end{array}
\end{equation}
A dipole-vertex is said to be separating if its removal disconnects $\cG$ into two graphs $\cG_1, \cG_2$ and \emph{non-separating} otherwise, and similarly for chain-vertices. For a quantity $\mathcal{O}$, let us denote $\Delta\mathcal{O} = \mathcal{O}(\cG_1) + \mathcal{O}(\cG_2) - \mathcal{O}(\cG)$ in the case of a separating removal and $\Delta \mathcal{O} = \mathcal{O}(\cG') - \mathcal{O}(\cG)$ in the case of a non-separating removal, with $\cG'$ being the (connected) graph obtained after the removal.

\paragraph{Separating dipole removals.} It can be checked that for all types of dipoles
\begin{equation}
    h(\cG) = h(\cG_1) + h(\cG_2),\qquad l(\cG) = l(\cG_1) + l(\cG_2)
\end{equation}
This is proved in \cite{BeCa} for tetrahedral interactions. Let us give an example in the case of a dipole pillow of color 1,
\begin{equation}
    \begin{array}{c} \includegraphics[scale=.3]{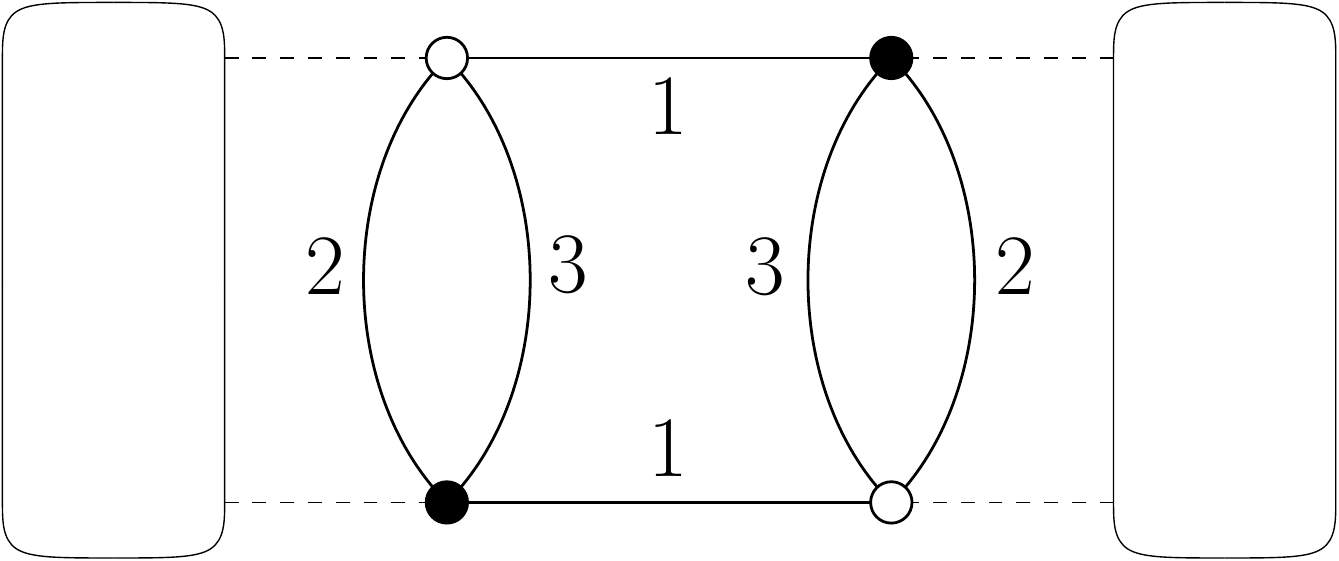} \end{array} \to \begin{array}{c} \includegraphics[scale=.3]{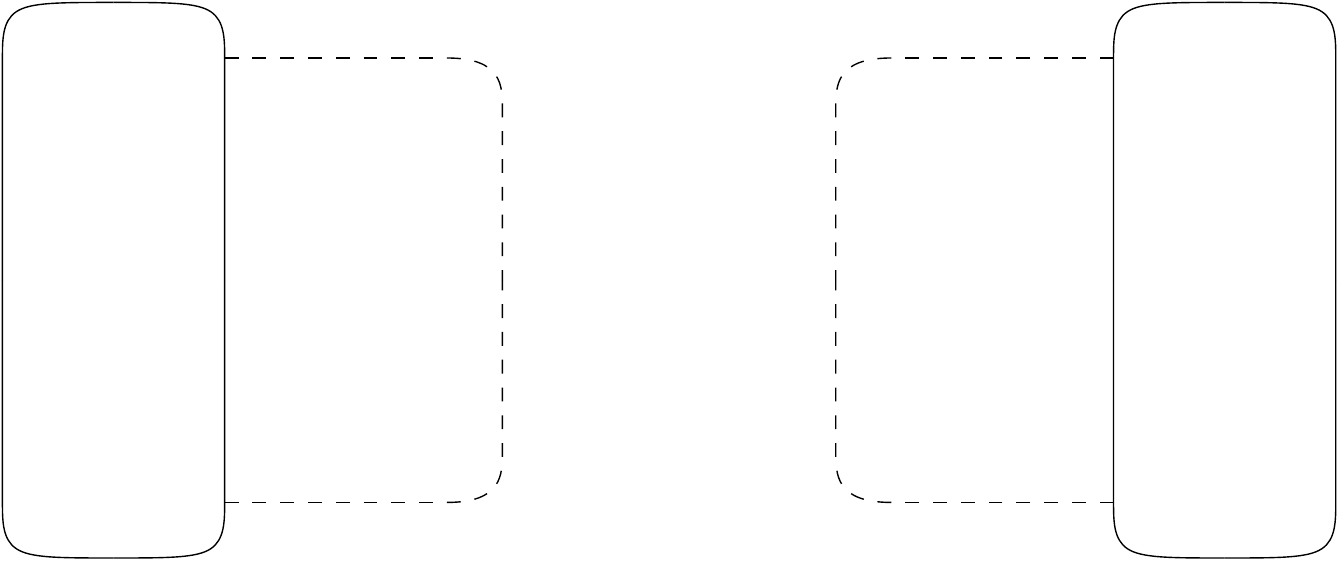} \end{array}
\end{equation}
Then, it is clear that 
\begin{equation} \label{VariationsPillowColor1Removal}
    \Delta F_{02} = \Delta F_{03} = 0, \Delta n_1=-1, \Delta E_0 = -2, \Delta F_{01} = 1.
\end{equation}
Then one finds $\Delta h = \Delta l =0$ from \eqref{GenusExpansion} and \eqref{Grade} (which hold for each connected component independently).

\paragraph{Non-separating dipole removals.}
\subparagraph{Non-separating dipoles of color $a\in\{1,2\}$.} Then
\begin{equation} \label{DipoleRemoval12}
    \Delta h =-1\text{ or }0, \qquad \Delta l = -2 \text{ or } -4.
\end{equation}
The case of dipoles built on tetrahedral interactions is proved in \cite{BeCa}. Let us consider the case of a dipole made of a pillow of color 1. Then, the variations are the same as in \eqref{VariationsPillowColor1Removal} except for $\Delta F_{01} \geq -1$ since $\cG'$ can have either one more or one less face of color 1 than $\cG$. One concludes by taking the variations in \eqref{GenusExpansion} and \eqref{Grade} again.

\subparagraph{Non-separating dipoles of color 3.} Then 
\begin{equation} \label{DipoleRemoval3}
    \Delta h =-1, \qquad \Delta l = 0 \text{ or } -4.
\end{equation}
The case of dipoles of color 3 built on tetrahedral interactions is proved in \cite{BeCa}, therefore we only have to consider dipoles made of a pillow of color 3. Then it is found that
\begin{equation}
    \Delta F_{01} = \Delta F_{02} = 0, \Delta F_{03} = \pm1, \Delta E_0 = -2, \Delta (n_{3b} + n_{3nb}) = -1,
\end{equation}
since $\cG'$ can have either one more or one less cycle of colors $\{0,3\}$. One concludes with \eqref{GenusExpansion} and \eqref{Grade} once again.

\paragraph{Chain removals.}

A (maximal) chain removal can be performed by removing a dipole from the chain, then removing the melonic 2-point functions this creates. The removal of a separating chain is obviously the same as that of a separating dipole.

\subparagraph{Non-separating chains of color $i$.}
Chains of color $i \in \{1,2,3\}$ are sequences of isolated dipoles of color $i$. Thus their removal is exactly the same as the removal of a non-separating dipole of color $i$.

\subparagraph{Non-separating broken chains.} By definition a broken chain has at least two dipoles of different colors. If it has no dipoles of color 3, then it must be a special case of \eqref{DipoleRemoval12} and if it has one it must also be a special case of \eqref{DipoleRemoval3}. By checking all cases, it is found that
\begin{equation} \label{BrokenChainRemoval}
    \Delta h=-1,\qquad \Delta l=-4.
\end{equation}
This is proved for tetrahedral interactions in \cite{BeCa}. Let us give an example with a pillow bubble of color 1 and a pillow bubble of color 3. Then $\Delta F_1=\Delta F_3=0$ and $\Delta F_2 =-1$. Moreover $\Delta E_0=-4$ and $\Delta n_1=\Delta (n_{3b}+n_{3nb})=-1$. One concludes with \eqref{GenusExpansion} and \eqref{Grade}.

\subsubsection{Skeleton graph of a scheme}

Let $\cS$ be a scheme with genus $g$ and grade $l$. Similarly to what has been done for the $O(N)^3$ model in~\cite{Bonzom4}, we introduce the \emph{skeleton graph} $\mathcal{I}(\cS)$ of a scheme $\cS$.
\begin{itemize}
\item We call the \emph{components} of $\cS$ the connected components obtained after removing all chain-vertices of $\cS$. In each component, we mark the edges created by the removals.
\item The vertices of $\mathcal{I}(\cS)$ are the components of $\cS$.
\item Two vertices of $\mathcal{I}(\cS)$ are connected by an edge if the two components are connected by a chain-vertex in $\cS$. Each edge is labeled by the type of the chain connecting them.
\end{itemize}

\begin{lemma}\label{thm:SkeletonGraph}
The skeleton graph $\mathcal{I}(\cS)$ of a scheme $S$ satisfies the following properties:
\begin{enumerate}
\item Any vertex of $\mathcal{I}(\cS)$ corresponding to a component of vanishing genus and grade, and not carrying the external legs of $\cS$, has degree at least $3$.
\item If $\mathcal{I}(\cS)$ is a tree, then the genus and grade of $\cS$ are split among the components.
\end{enumerate}
\end{lemma}

\begin{proof}
\begin{enumerate}
\item If a component of vanishing genus and grade not carrying the external legs of $\cS$ has valency $1$ in $\mathcal{I}(\cS)$, then it is a melonic component, which is not possible for a scheme. Similarly if it has valency $2$ in $\mathcal{I}(\cS)$ then it is a chain. This implies that the (two) chains incident to this component are not maximal, which is not possible in a scheme.
\item As we have shown, removing a separating chain-vertex in $\cS$ splits the degree among the components of $\cS$. Saying that $\mathcal{I}(\cS)$ is a tree is equivalent to saying that all chain-vertices of $\cS$ are separating, therefore the genus and grade of $\cS$ split among the components.
\end{enumerate}
\end{proof}

\subsubsection{Identifying the dominant schemes}

We first prove that if $\mathcal{S}$ is a dominant scheme, its skeleton graph $\mathcal{I}(\mathcal{S})$ is a tree. Assume that $\mathcal{I}(\mathcal{S})$ is a not a tree, i.e. $\mathcal{S}$ has a non-separating chain.

\begin{itemize}
    \item If it has a non-separating broken chain-vertex, then removing it decreases the grade, see Equation \eqref{BrokenChainRemoval}, so $\mathcal{S}$ cannot be dominant. The same holds if it has a non-separating chain-vertex of color $a\in\{1,2\}$, see Equation \eqref{DipoleRemoval12}.
    \item If it has a non-separating chain or dipole of color 3, there two possibilities. Either the grade decreases, then $\mathcal{S}$ cannot be dominant, or it does not but the genus decreases anyway, see Equation \eqref{DipoleRemoval3}. We then perform the following move
    \begin{equation}
        \begin{array}{c}\includegraphics[scale=.4]{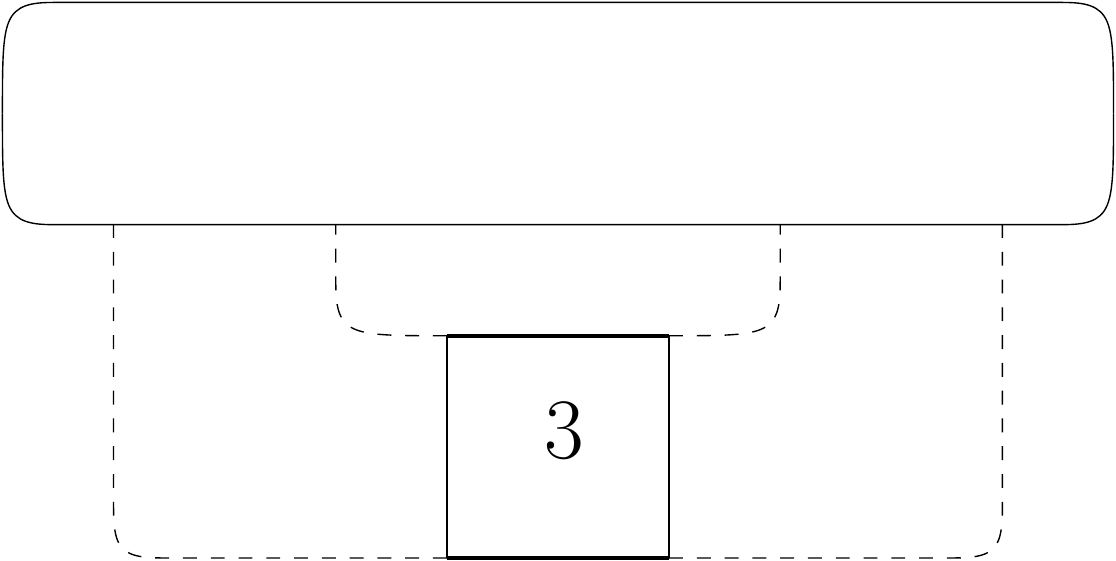}\end{array} \quad \to \quad \begin{array}{c}\includegraphics[scale=.4]{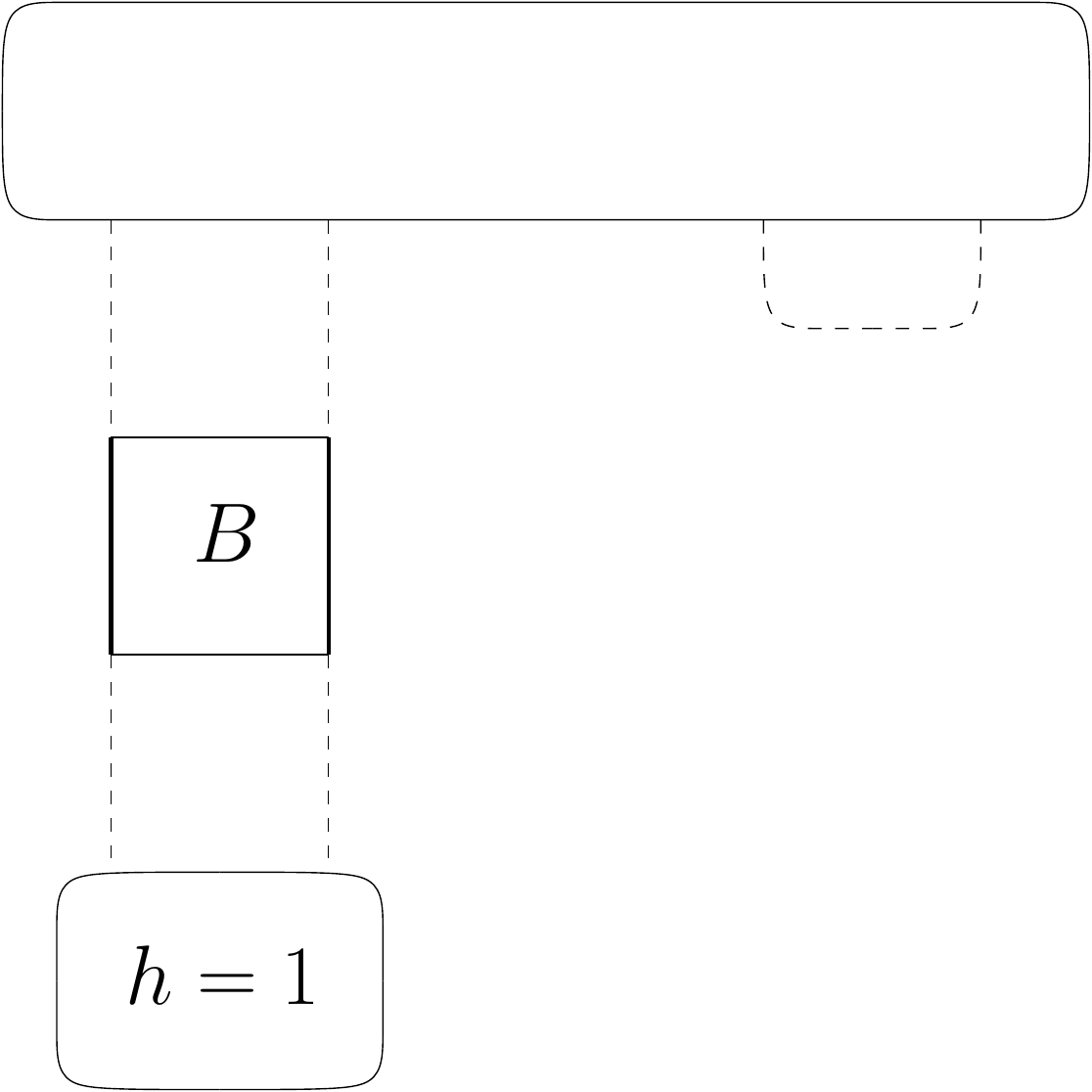}\end{array}
    \end{equation}
    where the component which is added has genus 1. On the RHS we get a scheme with the same genus as on the LHS but one more broken chain, hence it is more singular and as a result $\mathcal{S}$ cannot be dominant.
\end{itemize}

Thus, the skeleton graph must be a tree $\mathcal{T}$. Using Lemma~\ref{lem:UNxOD}, we know that the leaves of $\mathcal{T}$ cannot have vanishing genus and grade, therefore all leaves of $\mathcal{T}$ must have genus at least $1$ (and grade 0), hence there are at most $g$ leaves in $\mathcal{T}$. Finally, in order to be dominant, $\mathcal{T}$ must have as many broken chains as possible. We thus have an optimization problem: maximizing the number of chains with an upper bound on the number of leaves, and involving as variables the degrees and the genus of the internal vertices of $\mathcal{T}$. The solution is to have as many leaves as possible, here $g$, and inner vertices of degree exactly 3 so that $\mathcal{T}$ is a plane binary tree \cite{TaFu}.

In terms of schemes, the leaves of $\mathcal{T}$ correspond to 2-point schemes of genus 1 (and vanishing grade) with no separating chains, while the internal vertices of $\mathcal{T}$ correspond to 6-point functions of vanishing genus (and grade). By repeating the analysis of \cite{BeCa}, one finds the structure of those objects. This gives the following proposition.

\begin{prop} \label{prop:dom_scheme_un2}
A dominant scheme of genus $h>0$ has $2h-1$ broken chain-vertices, all separating. Such a scheme has the structure of a rooted\footnote{The root is a marked leaf.} binary plane tree where
\begin{itemize}
\item Edges correspond to broken chain-vertices.
\item The root of the tree corresponds to the two external legs of the 2-point function.
\item The $h$ leaves are one of the two following graphs
\begin{equation}
\includegraphics[scale=0.5]{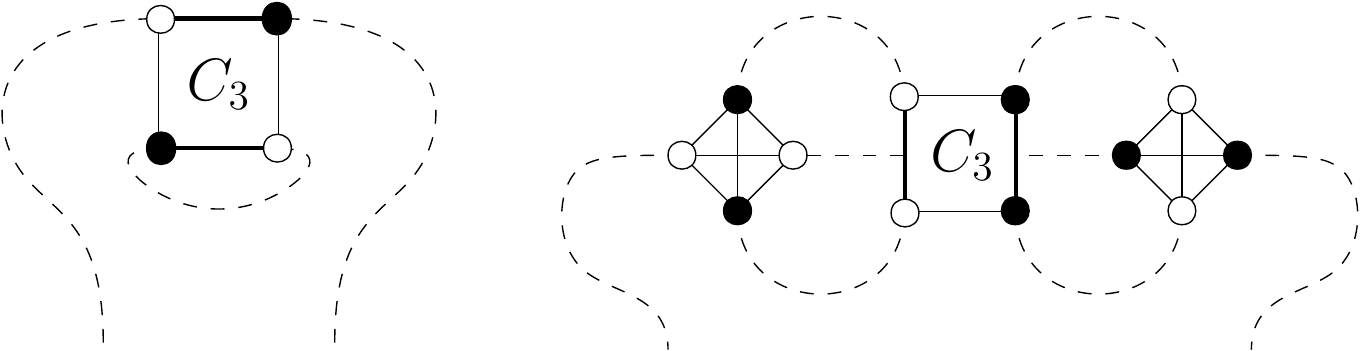}
\end{equation}
Observe that the chains of color $3$ have different boundary vertices in the two graphs.
\item Each internal vertex corresponds to one of four $6$-point subgraphs:
\begin{equation}
\includegraphics[scale=0.5]{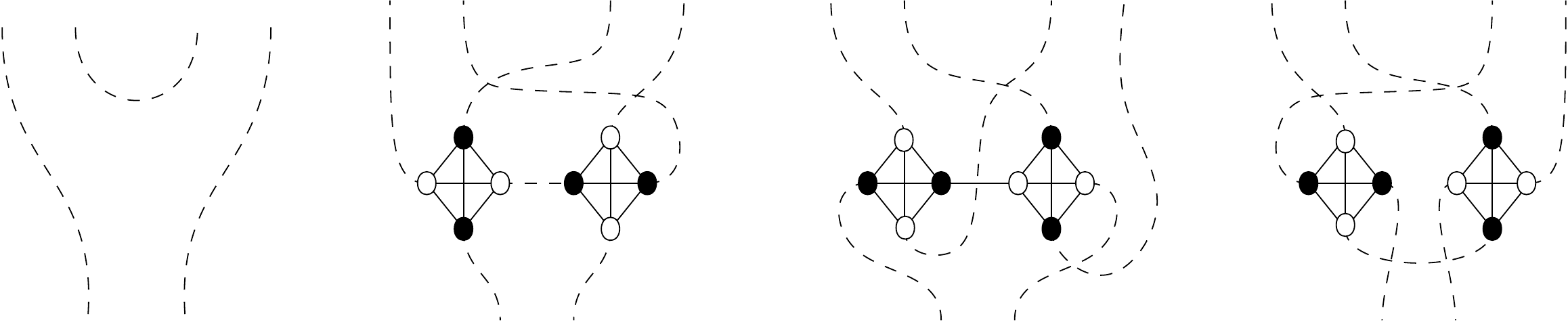}
\end{equation}
Here we have used the embedding convention that the edge coming out of the bottom left is connected to a white vertex.
\end{itemize}
\end{prop}

Notice that the two types of chains of color 3, which differ from their vertex colorings, play different roles (at the leaves of the tree). Remarkably, it is not necessary to also distinguish two types of broken chains according to their vertex colorings. Indeed, the only constraint imposed by the above Proposition is that every broken chain (corresponding to the edges of the tree) has one black and one white vertex on each side.

We now give the generating series of graphs associated to a dominant scheme. Recall that if $\mathcal{T}$ is a rooted  binary tree with $n$ edges, then it has $\frac{n-1}{2}$ internal vertices, $\frac{n+1}{2}$ leaves (not counting the root).

For a dominant scheme corresponding to a rooted plane binary tree $\mathcal{T}$, every edge of $\mathcal{T}$ contributes with the generating series of broken chains, every inner node with $(1+6t)$ (1 for the leftmost case of internal nodes of Proposition \ref{prop:dom_scheme_un2}, all 6 cases come with a factor $t$). Since at genus $h>0$, a dominant scheme has are $2h-1$ edges, its generating series is
\begin{align}
G_{\mathcal{T}}^{h}(t,\mu) = L^{2-2h}\bigl(C_{3,\includegraphics[scale=0.2,valign=c]{wbbw.pdf}}(t,\mu)+tC_{3,\includegraphics[scale=0.2,valign=c]{wwbb.pdf}}(t,\mu)\bigr)^{h}\left(1+6t\right)^{h-1}B(t,\mu)^{2h-1}
\end{align}
where $B(t,\mu)$ is given by Equation~\eqref{BrokenChainsUN2OD}. We recall that $L=N/\sqrt{D}$ is the parameter of the genus expansion, see Theorem \ref{thm:free_energy}.

Since it only depends on $h$, the sum over all rooted binary trees with $h$ leaves (not counting the root) is 
\begin{equation}
G_{\text{dom}}^{h}(t,\mu) = \sum_{\substack{\mathcal{T}\\ \text{$h$ leaves}}} G_{\mathcal{T}}^{h}(t,\mu) = L^{2-2h}\frac{\Cat_{h-1}}{(1+6t) B(t,\mu)} \Bigl(\bigl(C_{3,\includegraphics[scale=0.2,valign=c]{wbbw.pdf}}+tC_{3,\includegraphics[scale=0.2,valign=c]{wwbb.pdf}}\bigr)(1+6t)B(t,\mu)^{2}\Bigr)^h.
\end{equation}
Here $\Cat_{h-1} = \frac{1}{h}\binom{2h-2}{h-1}$ the number of rooted trees with $h$ leaves. $G_{\text{dom}}^{h}(t,\mu)$ is the total contribution of dominant schemes of genus $h$.

\subsubsection{Double scaling limit for the 2-point function}

As higher order graphs in the genus expansion can have more broken chains, we can tune the way we approach the critical point such that the singular behaviour of broken chains makes up for the loss of scaling in $L$ of the graphs, enhancing their contribution. Thus the double scaling parameter $\kappa(\mu)$ is defined graphs of all genus contribute. The dominant schemes of genus $h$ scale with $L$ as $L^{2-2h}$, therefore we defined $\kappa(\mu)$ as
\begin{equation}
\kappa(\mu)^{-1} = L^2 \frac{1}{(1+6t)(C_{3,\includegraphics[scale=0.2,valign=c]{wbbw.pdf}}(t,\mu)+tC_{3,\includegraphics[scale=0.2,valign=c]{wwbb.pdf}}(t,\mu))}\left(\frac{1}{B(t,\mu)}\right)^2
\label{eq:kappa_un2}.
\end{equation}
Using Equation~\eqref{eq:kappa_un2} we have
\begin{equation}
 \frac{1}{(1 + 6t)B(t,\mu)}  = \frac{\kappa(\mu)^{-\frac{1}{2}}}{L} \left(\frac{C_{3,\includegraphics[scale=0.2,valign=c]{wbbw.pdf}}+tC_{3,\includegraphics[scale=0.2,valign=c]{wwbb.pdf}}}{1+6t}\right)^{\frac{1}{2}}
\end{equation}

Therefore, the contribution of the dominant scheme of genus $h>0$ in the double scaling limit reads
\begin{equation}
G_{dom}^{h}(t,\mu) = L\Cat_{h}\left(\frac{(C_{3,\includegraphics[scale=0.2,valign=c]{wbbw.pdf}}+tC_{3,\includegraphics[scale=0.2,valign=c]{wwbb.pdf}})}{1+3t}\bigg\rvert_{(t_c(\mu),\mu)} \right)^{\frac{1}{2}}\kappa(\mu)^{h-\frac{1}{2}}
\end{equation}.

For $h=0$, there is only one scheme with one edge and no vertices, corresponding to the generating function of melonic graphs. It contributes as $M_c(t_c(\mu),\mu)$.

The leading order of the $2$-point function in the double scaling limit is obtained by resumming the contribution of all genus $h$. Therefore we get
\begin{align}
G_{2}^{DS}(\mu) &= \sum\limits_{h\in\mathbb{N}} G_{dom}^h(\mu) \\
& = M_c(t_c(\mu),\mu)+ L M_c(t_c(\mu),\mu) \sum\limits_{h \geq 1}^{} \Cat_{h-1}\left(\frac{(C_{3,\includegraphics[scale=0.2,valign=c]{wbbw.pdf}}+tC_{3,\includegraphics[scale=0.2,valign=c]{wwbb.pdf}})}{1+6t}\bigg\rvert_{(t_c(\mu),\mu)} \right)^{\frac{1}{2}}\kappa(\mu)^{h-\frac{1}{2}} \nonumber\\
&= M_c(t_c(\mu),\mu)\left(1+ L\kappa(\mu)^{\frac{1}{2}} \left(\frac{(C_{3,\includegraphics[scale=0.2,valign=c]{wbbw.pdf}}+tC_{3,\includegraphics[scale=0.2,valign=c]{wwbb.pdf}})}{1+6t}\bigg\rvert_{(t_c(\mu),\mu)} \right)^{\frac{1}{2}} \sum\limits_{h}^{} \Cat_{h} \kappa(\mu)^h\right) \nonumber\\
&= M_c(t_c(\mu),\mu)\left(1+ L \left(\frac{(C_{3,\includegraphics[scale=0.2,valign=c]{wbbw.pdf}}+tC_{3,\includegraphics[scale=0.2,valign=c]{wwbb.pdf}})}{1+6t}\bigg\rvert_{(t_c(\mu),\mu)} \right)^{\frac{1}{2}} \frac{1-\sqrt{1-4\kappa(\mu)}}{2\kappa(\mu)^\frac{1}{2}}\right) \label{eq:GDS_U(N)2}
\end{align}
The sum converges for $\kappa(\mu) \leq \frac{1}{4}$, similarly to other tensor models. The parameter $\kappa(\mu)$ encodes a balance between the large $N$ limit and the criticality (in the large $D$ limit). In particular, for $\kappa(\mu) = 0$, we obtain the usual large $N,D$ limit where only melons contribute, and at the other end when $\kappa(\mu) = \frac{1}{4}$ analyticity is lost. 

This function has a square-root singularity, just like the double scaling 2-point function of tensor graphs in \cite{GuSch}, as well as the one in the multi-orientable model \cite{GuTaYo} and the $O(N)^3$-invariant model \cite{Bonzom4}, and also similar to that of the $U(N)^2\times O(D)$-invariant model with tetrahedral interaction \cite{BeCa}. 

Note however that the scaling with $L$ differs from that of \cite{Bonzom4}. This is due to differences in the dominant schemes between the two models (although both are indeed mapped to rooted binary plane trees). In the $U(N)^2 \times O(D)$-invariant multi-matrix model, the leaves of the trees associated to a dominant scheme all have genus $1$ and vanishing grade. However, they have degree $\omega = \frac{1}{2}$ in the $O(N)^3$-invariant tensor model. Thus, the dominant schemes of the $U(N)^2 \times O(D)$-invariant multi-matrix models are not mapped to the dominant schemes of the $O(N)^3$ tensor model by the map $\theta$ defined in Equation~\eqref{eq:theta}.

\section{\label{sec:UNODMM} The \texorpdfstring{$U(N) \times O(D)$}{U(N)XO(D)} bipartite multi-matrix model with tetrahedral interaction}

\subsection{\label{UNOD:mel_chain} Definition of the model and its large \texorpdfstring{$N$}{N}, large \texorpdfstring{$D$}{D} expansion}

\subsubsection{\label{UNOD:graph_exp} Feynman graphs, genus and grade}

\paragraph{The model.\\}

The $U(N)\times O(D)$ multi-matrix model is a close cousin of the $U(N)^2\times O(D)$ tensor model. We consider a single quartic interaction and its complex conjugate,
\begin{equation} \label{UNODAction}
S_{U(N)\times O(D)}(X_\mu, X_\mu^\dagger) = -ND \sum\limits_{\mu = 1}^{D} \Tr(X_{\mu}^{\dagger} X_{\mu}) + \frac{\lambda}{4} ND^{\frac{3}{2}} \sum_{\mu, \nu=1}^D \Bigl(\Tr (X_\mu X_\nu X_\mu X_\nu) + \Tr (X^\dagger_\mu X^\dagger_\nu X^\dagger_\mu X^\dagger_\nu) \Bigr)
\end{equation}
It is invariant under 
\begin{equation}
X_{\mu} \mapsto \sum_{\mu'=1}^D O_{\mu\mu'}\ UX_{\mu'}U^\dagger
\end{equation}
where $O\in O(D)$ and $U\in U(N)$.

The Feynman expansion is performed using a different graphical representation than in the previous model, which is closer to matrix models. 
\begin{itemize}
\item We represent an interaction $\sum_{\mu, \nu} \Tr (X_\mu X_\nu X_\mu X_\nu)$ as a white vertex of degree 4, and each matrix as an incident half-edge. These half-edges are cyclically ordered, say counter-clockwise, around each vertex, reflecting the structure of the trace.
\item The interaction $\sum_{\mu, \nu} \Tr (X_\mu^\dagger X_\nu^\dagger X_\mu^\dagger X_\nu^\dagger)$ is represented as a black vertex, also with 4 incident, cyclically ordered, half-edges.
\item Propagators glue half-edges incident to black vertices to half-edges incident to white vertices.
\end{itemize}
Feynman graphs are thus \emph{combinatorial maps} (i.e. ribbon graphs), which are bipartite and tetravalent. The faces of the map are the cycles obtained by following edges from black to white vertices and corners at vertices counter-clockwise (a corner is the portion between two consecutive edges at a vertex).

The propagator is $\frac{1}{ND} \delta_{ad} \delta_{bc} \delta_{\mu\nu}$. It therefore identifies the matrix indices of half-edges and as usual in matrix models, each face receives a weight $N$. Moreover, propagators also identify the vector indices of half-edges. There is therefore another notion of cycles which receive the weight $D$. They are cycles which follow edges and cross the vertices (i.e. leaving one half-edges on each side). There are similar to the bicolored cycles of color $\{0,3\}$ in the previous model, which also received the weight $D$, and we will call them \emph{straight cycles}.

Let $\bar{\mathbb{M}}$ be the set of vacuum, connected Feynman graphs. For $\bar{\cM}\in\bar{\mathbb{M}}$ we denote $V$, $E$, $F$, $\phi$ the numbers of vertices, edges, faces and straight cycles. 

An important result of \cite{TaFe} is the $1/N$ and $1/D$ expansion of this model.

\begin{theorem}[\cite{TaFe}]\label{thm:free_energy_U(N)}
The free energy admits the following expansion on maps
\begin{equation}
F_{U(N)\times O(D)}(\lambda) = \ln \int \prod_{\mu=1}^D dX_\mu dX_\mu^\dagger\ e^{S_{U(N)\times O(D)}(X_\mu, X_\mu^\dagger)} = \sum_{\bar{\cM}\in\bar{\mathcal{M}}} \biggl(\frac{N}{\sqrt{D}}\biggr)^{2-2g(\bar{\cM})} D^{2-l(\bar{\cM})/2}\ \mathcal{A}_{\bar{\cM}}(\lambda)
\end{equation}
where $g(\bar{\cM})$ is the genus of $\bar{\cM}$ and $l(\bar{\cM})$ is a non-negative integer called the \emph{grade},
\begin{equation} \label{GenusAndGradeU(N)}
\begin{aligned}
2-2g(\bar{\cM}) &= F-E+V,\\
\frac{l(\bar{\cM})}{2} &= 1 + g(\bar{\cM}) + E - \frac{3V}{2} - \phi.
\end{aligned}
\end{equation}
\end{theorem}
Like in the previous model, the large $N$, large $D$ expansion is in fact an expansion in $D$ and $L = \frac{N}{\sqrt{D}}$.


\paragraph{2-point maps.\\}

Due to the symmetries of the model, the 2-point function has the form
\begin{equation}
    \langle (X_\mu)_{ab} (X^\dagger_\nu)_{cd}\rangle = \frac{1}{N^2D} G_{N, D}(\lambda) \delta_{\mu\nu} \delta_{ad} \delta_{bc},
\end{equation}
where $G_{N, D}(\lambda) = \langle \sum_{\mu=1}^D \Tr X_\mu X^\dagger_\mu\rangle$. It has an expansion on 2-point maps, whose set is denoted $\mathbb{M}^\circ$. A 2-point map is like a vacuum map with exactly one half-edge left incident on a black vertex and one half-edge incident to a white vertex. These two half-edges are necessarily part of the same straight path, hence the factor $\delta_{\mu\nu}$ above.
From a 2-point map $\cM^\circ\in\mathbb{M}^\circ$ one can obtain a unique vacuum graph $\bar{\cM}\in \bar{\mathbb{M}}$ called its \emph{closure}, by connecting the two half-edges. This vacuum graph is furthermore equipped with a marked edge (that obtained by connecting the two half-edges) called a \emph{root}. The set of rooted maps, denoted $\mathbb{M}$, is in fact the set of Feynman graphs for the expansion of $G_{N,D}(\lambda)$.
Note that in contrast with the model of the previous section where we had defined equivalent objects, it is here not true that the amplitudes of $\cM^\circ\in\mathbb{M}^\circ$ and of its closure $\bar{\cM}$ are simply related by a factor $N^2 D$. This is because in the previous model the faces on both sides of an edge could never be the same (they carry different colors) while here they may. As a consequence, cutting an edge of  $\cM$ might break either one or two faces. 

However, it is clear that the amplitude of a rooted map is the same as that of the vacuum map obtained by removing the marking on the root. Although there is a clear bijection between 2-point maps and rooted maps, the expansions are thus a bit different. In the following we will focus on the rooted function $G_{N,D}(\lambda)$.

If $\cM$ is a rooted map and $\bar{\cM}\in\bar{\mathbb{M}}$ the one obtained by removing the marking, we define $h (\cM):= h(\bar{\cM})$ and $l(\cM) := l(\bar{\cM})$.

\subsubsection{Melons, dipoles and chains}

\paragraph{Melons.\\}

The elementary melon is the following 2-point map,
\begin{equation}
\includegraphics[scale=.45, valign=c]{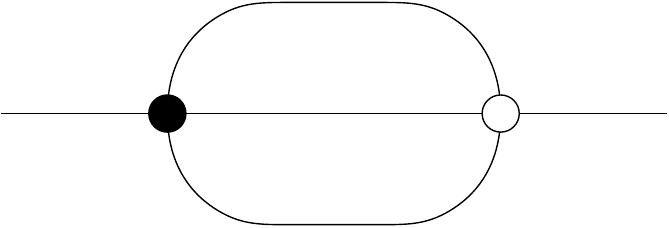}
\end{equation}
Melonic maps are defined starting from the elementary melon itself, and replacing any edge (or one of the half-edges) with another elementary melon, and so on. After inserting some elementary melons, a rooted map is obtained by gluing the two half-edges and marking this edge.



\begin{prop}
Melonic maps are the only maps with vanishing genus and grade.
\label{claim_mel}
\end{prop}
This proposition will be proved below.

The generating function $M(t)$ of melonic maps satisfies
\begin{equation}
M(t) = 1 +tM(t)^4
\label{eq:mel_unxod}
\end{equation}
where $t$ counts the number of melonic insertions in the map (equivalently half the number of vertices). The singular points and values of this function have been studied in~\cite{GuSch}. In particular its leading singularity is at $t_c=\frac{3^3}{4^4}$ where it has value $M(t_c)=\frac{4}{3}$.

\paragraph{Dipoles.\\}

Dipoles are $4$-point functions obtained from the elementary melon by cutting an edge. The $4$ half-edges thus obtained form two pairs of half-edges: one pair formed by the original half-edges of the elementary melon, and one pair from the cutting of the edge to get a dipole. There are two different dipoles depending on the edge which is cut.
\begin{itemize}
\item Cutting the ``top'' or ``bottom'' edge of the elementary melon gives a dipole with a face of degree $2$, called a \emph{U-dipole}
\begin{equation}
\includegraphics[scale=.5, valign=c]{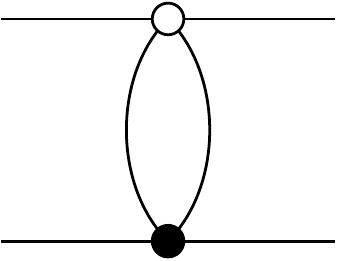}
\end{equation}
\item Cutting the ``middle'' edge of the elementary melon gives a dipole with a cycle of length $2$, called an \emph{O-dipole}
\begin{equation}
\includegraphics[scale=.5, valign=c]{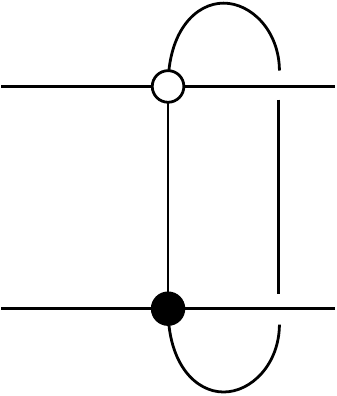}
\end{equation}
\end{itemize}
Notice that a U-dipole (respectively an O-dipole) is equivalent to a face of degree 2 (respectively a straight cycle of degree 2). We usually represent dipoles by drawing the half-edges of the same pair on the same ``side''.

Dipoles which are vertex-disjoint from other dipoles are said to be \emph{isolated}, else they are \emph{non-isolated}. A non-isolated dipole can in fact be part of only two possible minimal subgraphs, given in Figure \ref{fig:VertexJointDipoles}.

\begin{figure}
    \centering
    \includegraphics[scale=.5, valign=c]{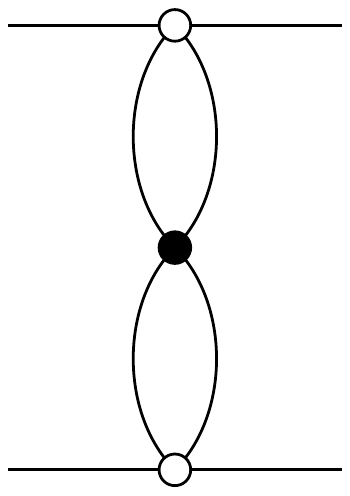}\hspace{2cm}\includegraphics[scale=.5,valign=c]{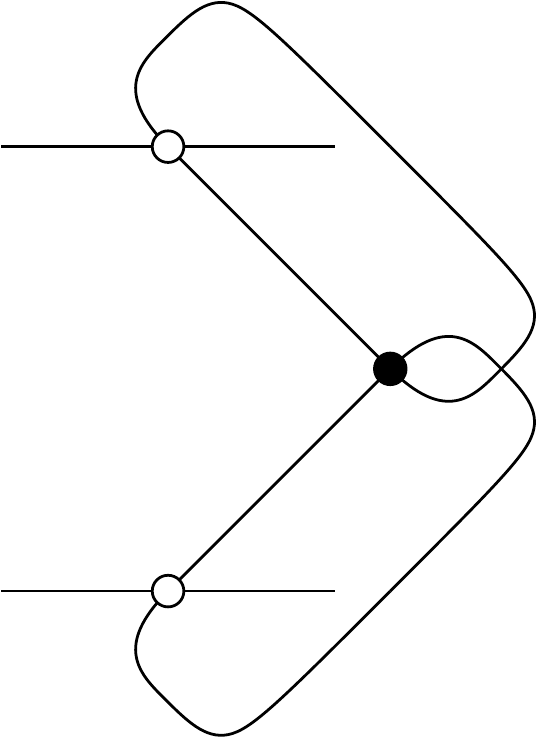}
    \caption{There are two possible minimal subgraphs containing non-isolated dipoles, up to vertex coloring.}
    \label{fig:VertexJointDipoles}
\end{figure}

Let us denote $D_O$ and $D_U$ the generating functions of O- and U-dipoles decorated with melons on its two (internal) edges and on two half-edges of a side (it does not matter which side; there is actually no way to distinguish them). We have
\begin{align}
D_O(t) &= U(t) = tM(t)^4 \underset{\eqref{eq:mel_unxod}}{=} M(t)-1 \\
D_U(t) &= 2U(t),
\end{align}
where the 2 in the last line comes from the two possible edges of the elementary melon which can be cut to obtain the U-dipole.

The \emph{dipole removal} deletes a dipole and connects the half-edges of the same side. This move is represented on Figure~\ref{fig:dip_rem_unxod}. It is straightforward to see that dipole removals cannot increase the genus nor the grade.

\begin{figure}[!ht]
\begin{center}
\includegraphics[scale=0.6]{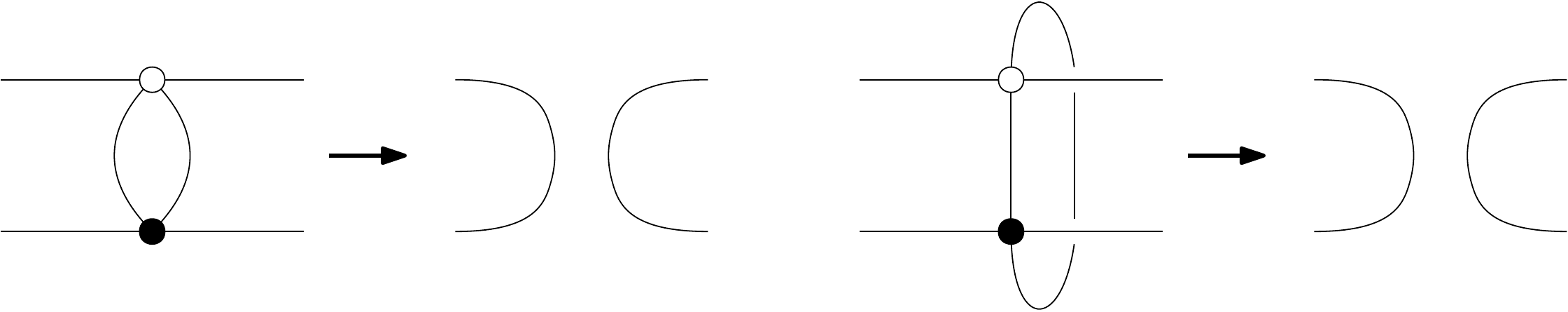}
\caption{On the left is the dipole removal of a U-dipole and on the right, of an O-dipole.}
\label{fig:dip_rem_unxod}
\end{center}
\end{figure}

\paragraph{Proof of Proposition~\ref{claim_mel}.\\}

First it is easy to prove, by induction on the number of vertices, that melonic maps all have vanishing genus and grade, by studying the melonic insertion,
\begin{equation}
\includegraphics[scale=0.45, valign=c]{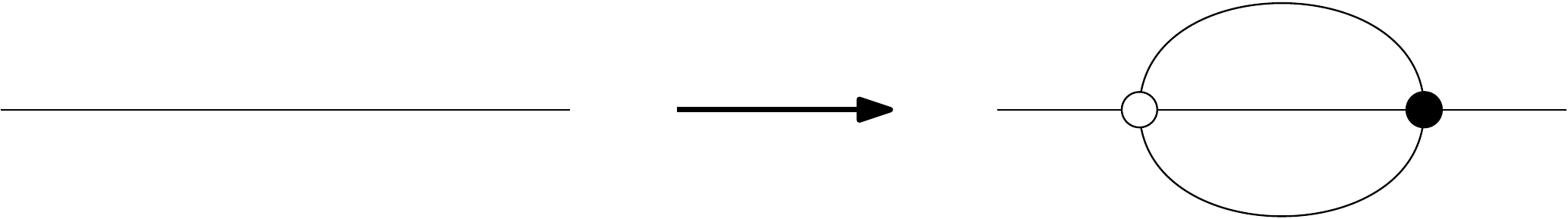}
\label{fig:melon_unxod}
\end{equation}

To prove that those are the only such maps, we also use an induction on the number of vertices $v$ of the map. It is obviously true at $v=2$. Let $v'> 2$ and assume that for all $v < v'$, there are no maps with vanishing genus and grade which are not melonic.

Let $\phi_{2n}$ be the number of straight cycles of length $2n$ in a map $\cM$ of genus $g$ and grade $l$. We first show that $\phi_{2}>0$ for maps of vanishing genus and grade. In an arbitrary map, the total number of straight cycles is $\phi = \sum_{n\geq 1} \phi_{2n}$ and since each edge belongs to exactly one cycle one has $E = \sum_{n\geq 1} 2n\phi_{2n}$. Moreover, $V=E/2$ since vertices have degree 4. Plugging those relations into the second Equation of \eqref{GenusAndGradeU(N)} gives
\begin{equation}
\sum\limits_{n\geq 1} (n-2)\phi_{2n} = l - 2 -2g.
\label{eq:cycle_length}
\end{equation}
This shows that such a map $\cM$ has $\phi_2>0$, i.e. at least one O-dipole which we denote $c$. If $\cM$ with $v'$ vertices has vanishing genus, it is planar and thus $c$ splits the map in two regions,
\begin{equation}
\cM = \includegraphics[scale=.5, valign=c]{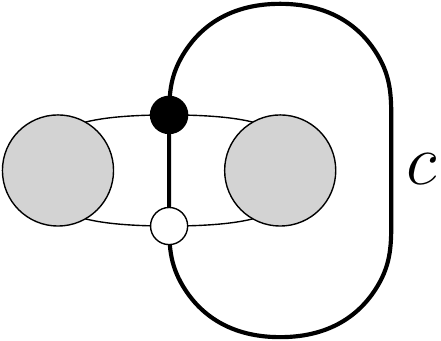}
\end{equation}
where each grey blob is {\it a priori} an arbitrary 2-point map. Removing the O-dipole disconnects the map into two connected components with fewer vertices, which are melonic by hypothesis. Reconstructing $\cM$ as above, we find that it is melonic.

\paragraph{Chains.\\}

A chain is either an isolated dipole, or a sequence of dipoles glued side by side, as in Equation~\eqref{Chain} where each dipole-vertex can now be a U-dipole or an O-dipole. The length of a chain is its number of dipoles. A chain of length $\ell$ contains chains of all lengths $1\leq \ell'\leq \ell$. A chain is said to be \emph{maximal} in $\cM$ if it is not contained in a longer chain. Note that maximal chains are \emph{vertex-disjoint}. As a remark, this would not be true if non-isolated dipoles were allowed as chains as in Figure \ref{fig:VertexJointDipoles}. As we have already pointed out in the previous model, chains in the literature are usually defined as having at least two dipoles. Here a single dipole can be a chain as long as it is vertex-disjoint from other dipoles.

We distinguish between chains made of O-dipoles only (resp. U-dipoles) named \emph{O-chains} (resp. \emph{U-chains}), which can have length greater than or equal to 1, and other chains which are called \emph{broken chains} and are made of at least two different dipoles. The following proposition is easily proved.

\begin{prop}
Changing the length of a chain does not change the genus nor the grade (both in vacuum and 2-point graphs).
\end{prop}

Let $C_O$, $C_U$ and $B$ be the generating functions of O-chains, U-chains and broken chains respectively, decorated with melons. They are given by
\begin{align}
C_O(t) &= \frac{U}{1-U} \\
C_U(t) &= \frac{2U}{1-2U} \\
B(t) &= \frac{3U}{1-3U} - \frac{U}{1-U} -\frac{2U}{1-2U}  = \frac{(4-6U)U^2}{(1-U)(1-2U)(1-3U)}
\end{align}
In the following, we will further need to distinguish between O-chains of even and odd lengths. Their respective generating functions are
\begin{equation}
C_{O,e}(t) = \frac{U^2}{1-U^2} \qquad
C_{O,o}(t) = \frac{U}{1-U^2}
\end{equation}

\subsubsection{Singularity analysis}

The singularity analysis is simpler than in the previous model, and is completely straightforward.

The leading singularity of $M(t)$ and of $U(t) = M(t)-1$ occurs at $t_c = \frac{3^3}{4^4}$. It is such that $U(t_c) = \frac{1}{3}$. The denominators of the series $C_O(t)$ and $C_U(t)$ remain finite at this point. However that of $B(t)$ blows up. Near $t_c$ it behaves like
\begin{equation}
B(t) \underset{t \rightarrow t_c}{\sim} \frac{1}{\sqrt{\frac{8}{27}}\sqrt{1-\frac{t}{t_c}}}.
\end{equation}
Therefore the ``most singular'' objects are the broken chains as they are the only ones to diverge at the critical point.

\subsubsection{Schemes}
\emph{Schemes} are maps obtained by
\begin{itemize}
\item Replacing every melonic submap with an edge,
\item Replacing every maximal chain with a \emph{chain-vertex}, i.e. a 4-point vertex where we forget the length of the chain. A chain-vertex keep track of the pairing of the half-edges by drawing them on the same side of a box, and separating both sides with fat edges along the box,
\begin{equation}
\includegraphics[scale=.5]{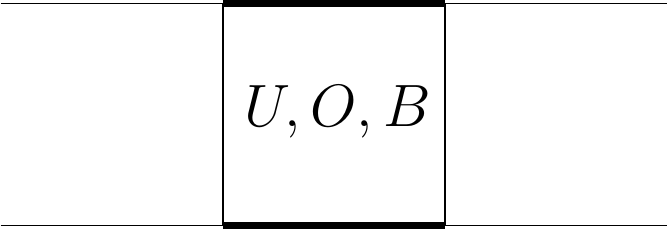}
\end{equation}
The labels $U,O,B$ correspond to chain-vertices replacing a U-chain, an O-chain or a broken chain.
\end{itemize}

From this definition, it is clear that the reduction of a map to its scheme does not change its genus and grade. It thus makes sense to define $\mathbb{MS}_{h,l}$ the set of schemes of genus $h$ and grade $l$. We will prove the following theorem.

\begin{theorem} \label{thm:SchemesUNxOD}
The set $\mathbb{MS}_{h,l}$ is finite.
\end{theorem}

\subsection{Finiteness of schemes of genus \texorpdfstring{$g$}{g} and grade \texorpdfstring{$l$}{l}}


\subsubsection{Removals of chains and dipoles}

We say that a dipole in $\cM$ is \emph{separating} if its removal disconnects $\cM$ into two connected maps $\cM_1, \cM_2$ and \emph{non-separating} otherwise, and similarly for chain-vertices.

\vspace{10pt}
\paragraph{Dipole removals.\\}

We consider the dipole removals represented on Figure~\ref{fig:dip_rem_unxod}. The following analysis parallels that made in the previous model, see Section \ref{sec:DipoleRemovals}.

\subparagraph{Separating dipole removal.} Denote $\Delta Q = Q(\cM_1) + Q(\cM_2) - Q(\cM)$ the variation of any quantity $Q$ through the removal. Recall that the formula \eqref{GenusAndGradeU(N)} holds for each connected map. During the removal one finds for both O- and U-dipoles that $\Delta F=\Delta \phi=0$, $\Delta E=-4$, $\Delta V=-2$ and the variation of the number of connected components is $\Delta C=1$. This gives
\begin{equation}
    \Delta h = \Delta l = 0.
\end{equation}

\subparagraph{Non-separating O-dipole removal.} Let $\Delta Q = Q(\cM') - Q(\cM)$ be the variation of $Q$ through the move. The number of faces is unchanged, i.e. $\Delta F=0$, while $\Delta \phi = 0,-2$ depending on the structure of the straight cycles which are incident to the dipole. This gives
\begin{equation} \label{ODipoleRemoval}
    \Delta h=-1, \qquad \Delta l = 0 \text{ or } -4.
\end{equation}
This is exactly like for the removal of a non-separating dipole of color 3 in the previous model, see Equation \eqref{DipoleRemoval3}.

\subparagraph{Non-separating U-dipole removal.} In contrast with the case of O-dipoles, it preserves the straight cycles incident to it, i.e. $\Delta \phi=0$, while $\Delta F=0,-2$ depending on the structure of the faces which are incident to the dipole. This gives
\begin{equation} \label{UDipoleRemoval}
    \Delta h = -1\text{ or }0, \qquad \Delta l = -2\text{ or }-4,
\end{equation}
just like for the removal of dipoles of colors 1 and 2 in the previous model, see Equation~\eqref{DipoleRemoval12}.

\paragraph{Chain removals.\\}

A chain (i.e. maximal chain or chain-vertex) removal is the following move
\begin{equation}
\includegraphics[scale=.56, valign=c]{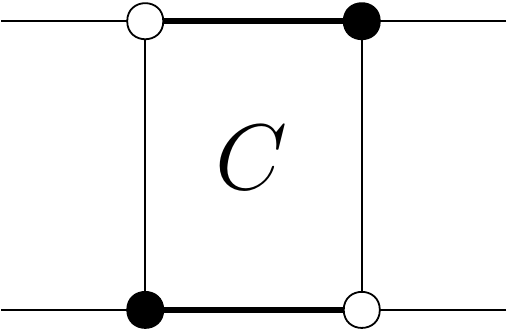} \qquad \to \qquad \includegraphics[scale=.55, valign=c]{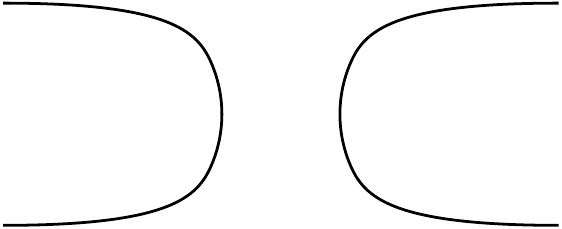}
\end{equation}
It can be studied by removing one dipole from the chain and then removing the resulting melonic 2-point functions. The removal of a separating chain is obviously the same as that of a separating dipole.

\subparagraph{Non-separating chains.}
O-chains and U-chains are sequences of isolated O- and U-dipoles. Thus their removal is exactly the same as the removal of a non-separating O- or U-dipole.

\subparagraph{Non-separating broken chains.} By definition a broken chain has at least one O-dipole and one U-dipole. Therefore the removal of a broken chain has to be a special case of both \eqref{ODipoleRemoval} and \eqref{UDipoleRemoval}. This directly gives
\begin{equation} \label{BrokenChainRemovalU(N)}
    \Delta h=-1,\qquad \Delta l=-4.
\end{equation}
This is the same result as for broken chain removals in the previous model \eqref{BrokenChainRemoval}.

\subsubsection{Skeleton graph of a scheme}

Let $\cS$ be a scheme with genus $g$ and grade $l$. Similarly to what has been done for the $O(N)^3$ model in~\cite{Bonzom4}, we introduce the \emph{skeleton graph} $\mathcal{I}(\cS)$ of a scheme $\cS$.
\begin{itemize}
\item We call the components of $\cS$ the connected components obtained after removing all chain-vertices of $\cS$. In each component, we mark the edges created by the removals.
\item The vertices of $\mathcal{I}(\cS)$ are the components of $\cS$.
\item Two vertices of $\mathcal{I}(\cS)$ are connected by an edge if the two components are connected by a chain-vertex in $\cS$. Each edge is labeled by the type of the chain connecting them.
\end{itemize}
An example of skeleton graph is shown on Figure~\ref{fig:skel_unxod}.

\begin{figure}
\centering
\includegraphics[scale=0.55]{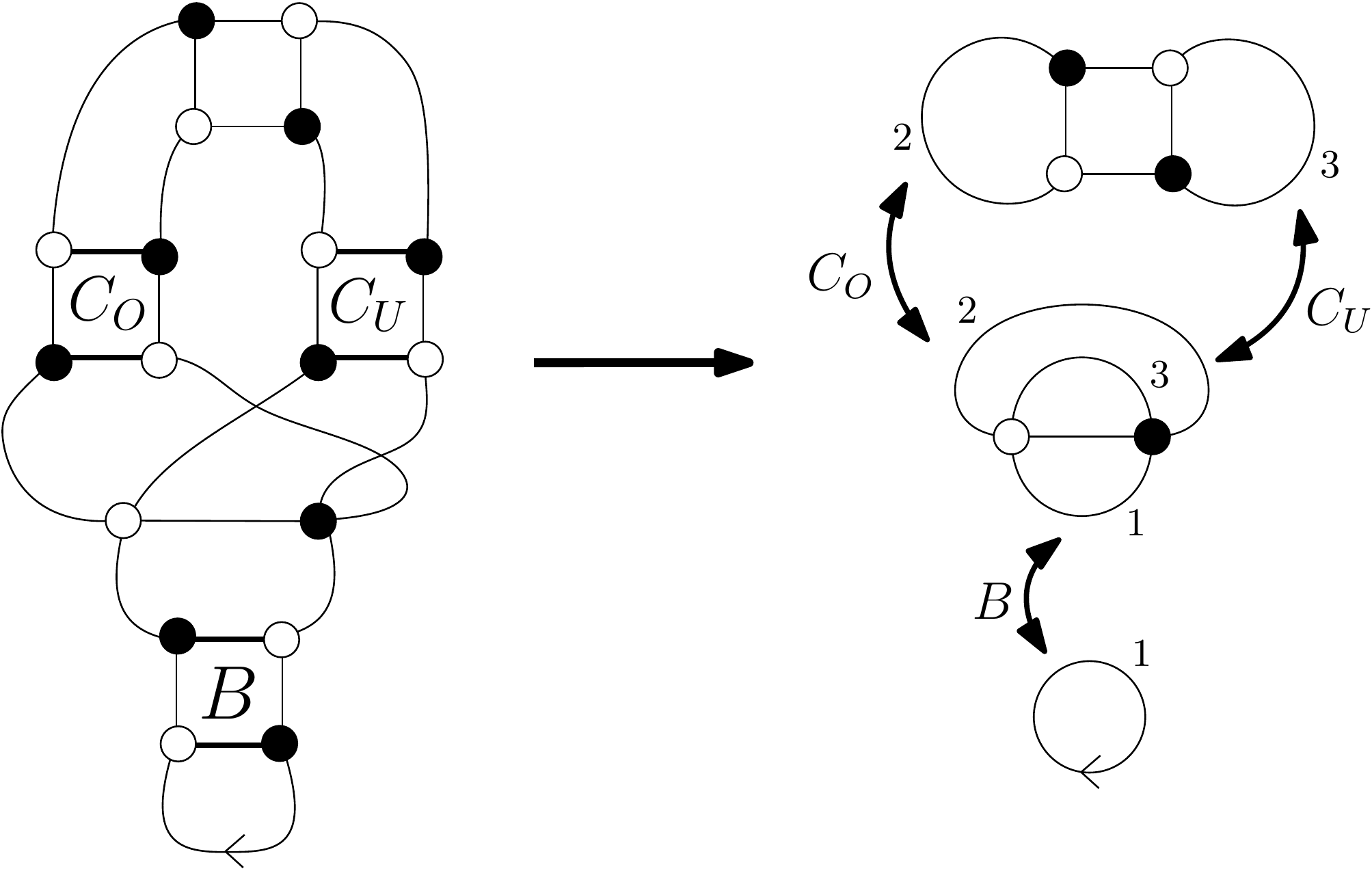}
\caption{A scheme of the $U(N) \times O(D)$ model and its skeleton graph.}
\label{fig:skel_unxod}
\end{figure}

\begin{lemma}\label{lem:UNxOD}
The skeleton graph $\mathcal{I}(\cS)$ of a scheme $S$ satisfies the following properties:
\begin{enumerate}
\item Any vertex of $\mathcal{I}(\cS)$ corresponding to a component of vanishing genus and grade, and not carrying the root edge of $\cS$, has degree at least $3$.
\item If $\mathcal{I}(\cS)$ is a tree, then the genus and grade of $\cS$ are split among the components.
\end{enumerate}
\end{lemma}
This is exactly the same as Lemma \ref{thm:SkeletonGraph} for the previous model.

\begin{proof}
\begin{enumerate}
\item If a component of vanishing genus and grade not carrying the root edge of $\cS$ has valency $1$ in $\mathcal{I}(\cS)$, then it is a melonic component, which is not possible for a scheme. Similarly if it has valency $2$ in $\mathcal{I}(\cS)$ then it is a chain. This implies that the (two) chains incident to this component are not maximal, which is not possible in a scheme.
\item As we have shown, removing a separating chain-vertex in $\cS$ splits the degree among the components of $\cS$. Saying that $\mathcal{I}(\cS)$ is a tree is equivalent to saying that all chain-vertices of $\cS$ are separating, therefore the genus and grade of $\cS$ split among the components.
\end{enumerate}
\end{proof}

\subsubsection{Finiteness of the number of schemes}

To prove Theorem~\ref{thm:SchemesUNxOD}, we will use the same 2-step strategy as in \cite{TaFu, Bonzom4}. We show that
\begin{enumerate}
    \item Schemes of genus $g$ and grade $l$ have finitely many chain-vertices.
    \item There are finitely many maps of genus $g$ and grade $l$ with $k$ isolated dipoles.
\end{enumerate}
This proves Theorem \ref{thm:SchemesUNxOD} by observing that there is a bijection between schemes of genus $g$ and grade $l$ and maps of the same genus and grade whose chains all have length 1. One can thus apply the result of the step 1, then replace chain-vertices with chains of length 1 and apply the result of step 2.

We start by showing that the number of chain-vertices is bounded by the genus and the grade. We have the following Lemma, the same as in \cite{BeCa}.

\begin{lemma}
If a scheme $\cS$ has a non-separating chain-vertex, then there exists a scheme $\cS'$ of the same genus and grade with more chain-vertices.
\label{lem:max_ns}
\end{lemma}

\begin{proof}
We sketch the proof and refer to \cite{BeCa} for details. Let $\mathcal{T}\subset \mathcal{I}(\cS)$ be a spanning tree in the skeleton graph of $\cS$. The edges in the complement $\mathcal{I}(\cS)\setminus\mathcal{T}$ correspond in $\cS$ to a set of non-separating chain-vertices. Removing each of them decreases the genus by 1 or the grade by 4. After all those removals, one obtains a scheme $\cS_r$ of genus $g-q_1$ and grade $l-4q_2$, such that $q_1+q_2$ is the number of non-separating chain-vertices which have been removed. It is then possible to attach to any (separating) remaining chain-vertex of $\cS_r$ a plane binary tree $\tilde{\mathcal{T}}$ with $q_1+q_2$ leaves and whose internal vertices correspond to components of vanishing genus and grade. Corresponding to the leaves of $\tilde{\mathcal{T}}$, one can choose, thanks to Lemma \ref{thm:SkeletonGraph}, components which make up for the loss of genus and grade induced by the removals of the non-separating chain-vertices. More formally, if $\ell$ is a leaf of $\tilde{\mathcal{T}}$, we can choose a component $\cM_\ell$ with genus $g_\ell$ and grade $\ell$, and using Lemma \ref{thm:SkeletonGraph}, such that $\sum_\ell g_\ell = q_1$ and $\sum_\ell l_\ell = q_2$. This gives a new skeleton graph which is a tree and which we denote $\mathcal{T}'$ and a (non-unique) scheme $\cS'$, such that $\mathcal{I}(\cS') = \mathcal{T}'$, whose degree and grade are the same as those of $\cS$ but has more chain-vertices.

As an illustration, we describe the attachment of $\tilde{\mathcal{T}}$ to $\mathcal{T}$ in the case $q_1+q_2=1$ (i.e. a single non-separating chain-vertex in $\cS$) on Figure~\ref{fig:non_sep_non_dom}.
\begin{figure}[!h]
    \centering
    \includegraphics[scale=0.6]{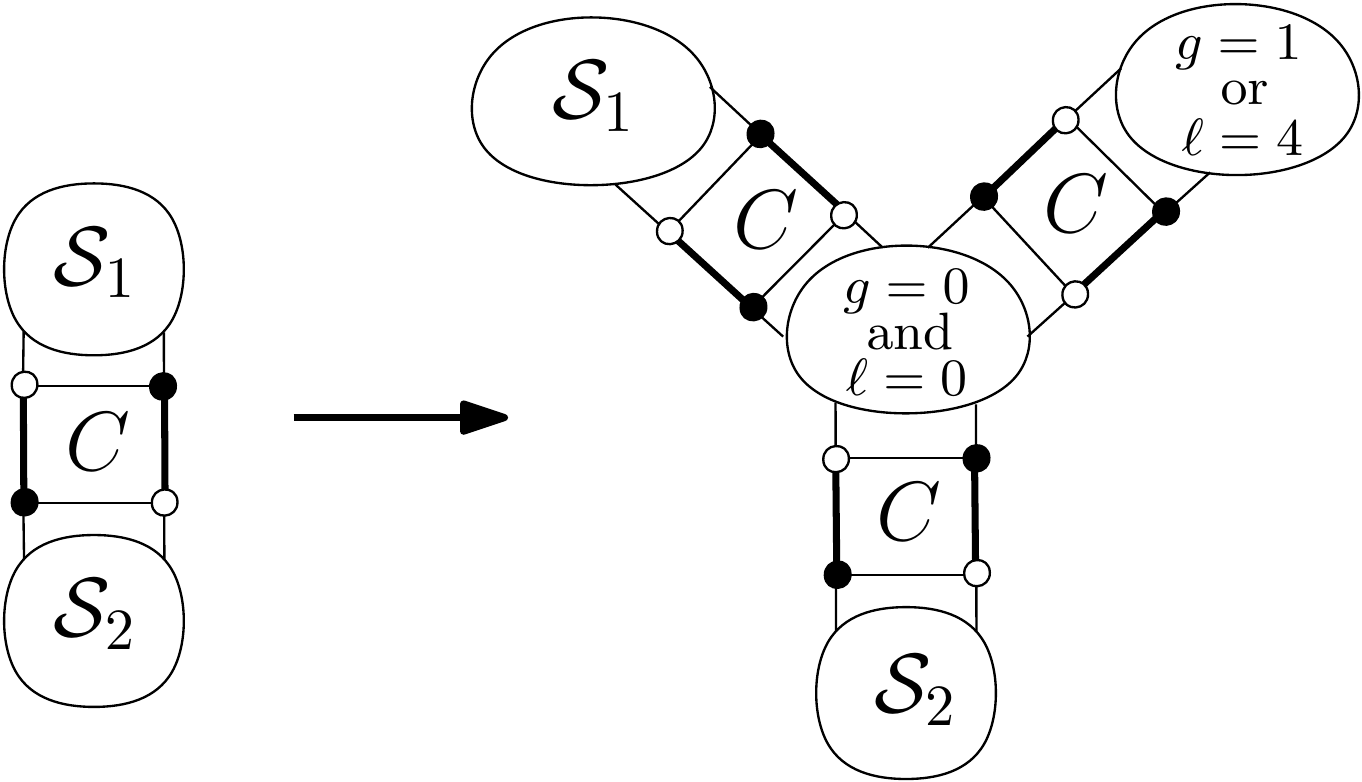}
    \caption{\label{fig:non_sep_non_dom} On the left is $\cS_r$ with a separating chain-vertex and a move attaching a binary vertex to it on the right hand side. It restores the genus and grade of $\cS$ but has more chain-vertices.}
\end{figure}
\end{proof}

From this lemma, we get the following bound on the number of chain-vertices of a scheme.
\begin{prop}
A scheme of genus $g$ and grade $l$ has at most $2(g+l)-1$ chain-vertices. This bounded is saturated when all chain-vertices are separating.
\label{prop:sch_chain}
\end{prop}

\begin{proof}
By contraposition of the Lemma~\ref{lem:max_ns}, schemes that have maximal number of chains cannot have non-separating chains. Therefore their skeleton graph must be a tree $\mathcal{T}$. Using Lemma ~\ref{lem:UNxOD}, we know that leaves of $\mathcal{T}$ cannot have vanishing genus and grade, therefore all leaves of $\mathcal{T}$ must have genus at least $1$ or grade at least $2$. Therefore, maximizing the number of chains in the scheme is a simple linear problem where one has to maximize the number of edges of a tree given a bound on the number of leaves given by $g+l$. The solution is to have as many leaves as possible, here $g+l$, and inner vertices of degree exactly 3 with vanishing genus and grade. A binary tree with $k$ leaves has $2k-1$ edges, thus the result follows.
\end{proof}

The second part of the proof of Theorem \ref{thm:SchemesUNxOD} comes from the following proposition.

\begin{prop}
The set of maps of genus $g$ and grade $l$ with $k$ isolated O-dipoles is finite.
\label{prop:sch_ns}
\end{prop}

\begin{proof}
The strategy is to show that there exist functions $b_{2n}(g,l,k)$ such that
\begin{equation}
    \phi_{2n} \leq b_{2n}(g,l,k)
\end{equation}
for all $n\geq 1$, where we recall that $\phi_{2n}$ is the number of straight cycles of length $2n$.

Observe that equation \eqref{eq:cycle_length}, i.e. $\sum_{n\geq 1} (n-2)\phi_{2n} = l - 2 -2g$, bounds the number of straight cycles of length greater than or equal to $6$ for fixed values of $g$ and $l$ and $\phi_2$. Therefore, two bounds remain to be obtained.
\begin{itemize}
    \item Bounding $\phi_2\leq b_2(h,l,k)$ in terms of the genus, grade and number of isolated O-dipoles. This implies bounds $\phi_{2n}\leq b_{2n}(h,l,k)$ for all $n\geq3$ thanks to \eqref{eq:cycle_length}.
    \item Bounding $\phi_4\leq b_4(h,l,k)$ similarly. This must be done independently because $\phi_4$ does not appear in \eqref{eq:cycle_length} and there could thus, {\it a priori}, exist maps with an arbitrarily large number of straight cycles of length 4.
\end{itemize}

\paragraph{Straight cycles of length 2\\}
Those cycles are O-dipoles by definition.
\begin{itemize}
    \item Either it is an isolated O-dipole and there are $k$ of them at most.
    \item Either it is a non-isolated O-dipole.
\end{itemize}
In the latter case, notice that a non-isolated dipole is non-separating, so its removal decreases the genus by $1$ and we can proceed by induction on the genus. If $\cM$ has genus zero (i.e. it is a planar map) then it cannot have non-isolated O-dipoles (they would form a topological minor of genus 1, see the right of Figure \ref{fig:VertexJointDipoles}). Hence we can define $b_2(0,l,k) = k$, which satisfies $\phi_2 \leq b_2(0,l,k)$ when the map $\cM$ is planar. Now if $\cM$ has genus $g>0$ and a non-isolated O-dipole, the latter can be removed and this decreases the genus by one and gives a connected map $\cM'$. This operation cannot increase the grade (in fact $l(\cM') = l$ or $l(\cM') = l-4$). It can form at most $2$ non-isolated O-dipoles (if the two new edges in $\cM'$ belong to non-isolated dipoles) or $1$ isolated O-dipole. Since $\cM'$ has at most $k+1$ isolated O-dipoles, we can define
\begin{equation}
b_2(g,l,k) = 2+\max_{\substack{k' \leq k+1\\ l'\leq l}} b_2(g-1,l',k')    
\end{equation}
By construction, this function satisfies $\phi_2 \leq b_2(g,l,k)$ for a map of genus $g$ and grade $l$. This shows that the number of O-dipoles is bounded in terms of the genus, grade and number of isolated O-dipoles.

This in turn proves that there is a finite number of straight cycles of length larger than or equal to 6.



Before treating the case of straight cycles of length 4, we give the following useful lemma (which could have in fact been used above to deal with cycles of length 2 too).

\begin{lemma} \label{thm:TopMinors}
Let $\mathfrak{m}$ be a connected topological minor of genus $h>0$. Then the number of copies of $\mathfrak{m}$ which can occur in $\cM$ is bounded linearly by the genus $g$ of $\cM$.
\end{lemma}

\begin{proof}
We say that a set of copies of $\mathfrak{m}$ are \emph{independent} in $\cM$ if they are vertex-disjoint in $\cM$, and it is maximal if it cannot be properly contained in another independent set.

Let $\operatorname{Indpt}_{\mathfrak{m}}(\cM)$ be a maximal independent set for $\mathfrak{m}$ in $\cM$. Since $h>0$ there can only be a finite number of disjoint copies of $\mathfrak{m}$, at most $\lfloor g/h\rfloor$, so $\lvert\operatorname{Indpt}_{\mathfrak{m}}(\cM)\rvert\leq g/h$. Any other occurrence of $\mathfrak{m}$ shares a vertex with one of $\operatorname{Indpt}_{\mathfrak{m}}(\cM)$. Since $\mathfrak{m}$ has a finite number of vertices $v(\mathfrak{m})$ and they have finite valency $4$, it comes that there is at most $4^{v(\mathfrak{m})-1} g/h$ copies of $\mathfrak{m}$ in $\cM$.
\end{proof}

We now treat the case of straight cycles of length 4. Let $\cC$ be such a cycle and let $V_{\text{exc}}$ the set of vertices which belong to at least one straight cycle of length 2 or 6 or larger.

\paragraph{$\cC$ has a vertex in $V_{\text{exc}}$\\}
If $\cC$ has a vertex which belongs to $V_{\text{exc}}$, then it has all its vertices at distance at most $2$ of $V_{\text{exc}}$. Since vertices all have valency $4$, for any vertex $v$ there are at most $17$ distinct vertices at distance at most $2$ of $v$. Thus there are at most $17|V_{\text{exc}}|$ straight cycles of length $4$ with a vertex in $V_{\text{exc}}$.

\paragraph{$\cC$ has no vertices in $V_{\text{exc}}$\\}
This means that $\cC$ only intersects straight cycles of length exactly $4$. In the following, we will assume $\cC$ is non-self-intersecting. The case where $\cC$ is self-intersecting can be tackled similarly. We denote  $v_1$ to $v_4$ the vertices of $\cC$ in any cyclic order. The cycle $\cC'$ incident to $v_1$ that is not $\cC$ also has length $4$. We distinguish cases depending on how $\cC$ and $\cC'$ intersect. 

\subparagraph{$\cC$ and $\cC'$ intersect only at $v_1$.} In this case, $\cC\cup\cC'$ forms a topological minor of genus $1$. From Lemma \ref{thm:TopMinors}, there is a bound in the genus on the number of such pairs, hence on the number of cycles of length $4$ which intersect another cycle of length $4$ exactly once. 

\subparagraph{$\cC$ and $\cC'$ intersect at least at $v_1$ and $v_2$.} By symmetry, it means they intersect at two vertices of distinct colors.
\begin{itemize}
    \item If they intersect at $v_1$ and $v_2$ only, then they either form an isolated U-dipole and there are at most $k$ of them, or a topological minor of genus 1 and we conclude with Lemma \ref{thm:TopMinors}. This is illustrated in Figure \ref{fig:v1v2}.
    \begin{figure}
        \centering
        \includegraphics[scale=.65]{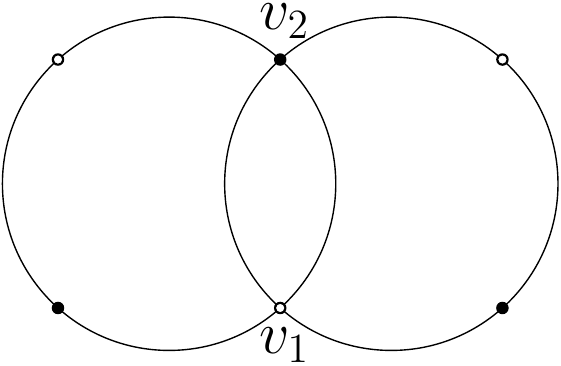}
        \hspace{1cm}
        \includegraphics[scale=.65]{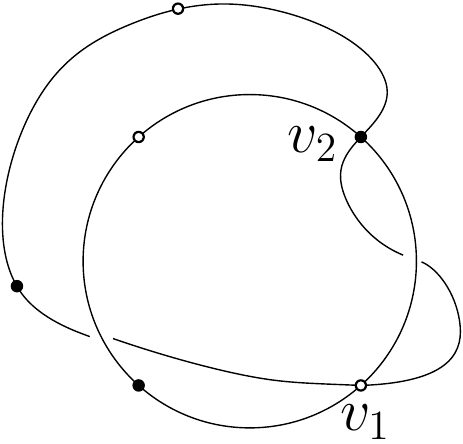}
        \caption{There are two possibilities for $\cC$ and $\cC'$ to intersect at exactly two vertices of different colors.}
        \label{fig:v1v2}
    \end{figure}
    \item If they intersect at $v_1$, $v_2$, $v_3$, or by symmetry at any triple of vertices, we have one of the situations depicted in Figure \ref{fig:v1v2v3}. In all of them, $\cC\cup \cC'$ forms a topological minor of genus 1. From Lemma \ref{thm:TopMinors}, we conclude that there is a bounded number of those minors.
    \begin{figure}
        \centering
        \includegraphics[scale=.65]{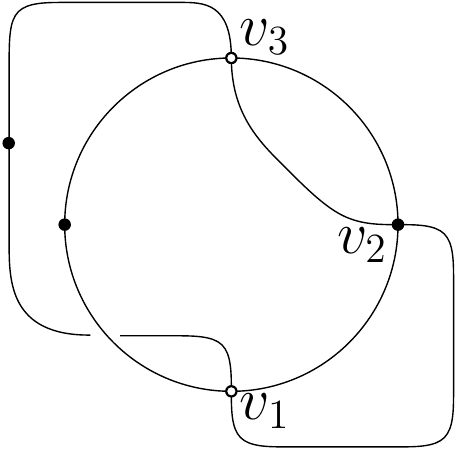}
        \hspace{1cm}
        \includegraphics[scale=.65]{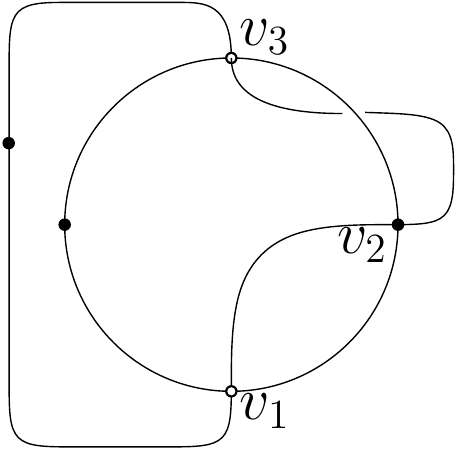}
        \hspace{1cm}
        \includegraphics[scale=.65]{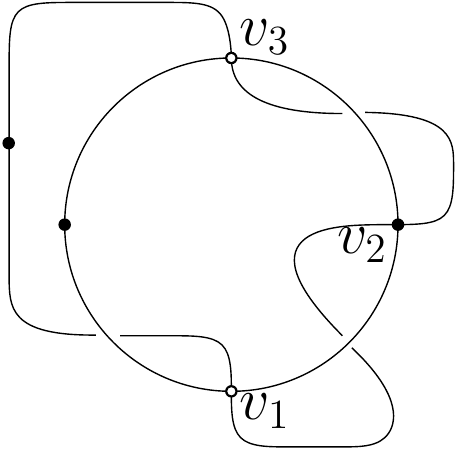}
        \caption{There are three possibilities for $\cC$ and $\cC'$ to intersect at exactly three vertices, up to symmetry.}
        \label{fig:v1v2v3}
    \end{figure}
    \item If they intersect on all four vertices, then only a finite number of maps can be drawn.
\end{itemize}

\subparagraph{$\cC$ and $\cC'$ intersect at $v_1$ and $v_3$ only.} By symmetry, it corresponds to the case where they intersect at exactly two vertices of the same color. We distinguish two cases depending on whether $\cC'$ is self-intersecting or not.
\begin{itemize}
\item If $\cC'$ is self-intersecting, we can check, see Figure \ref{fig:4nsi_with_si}, that $\cC\cup\cC'$ either forms a topological minor of genus 1 or of genus $2$. We conclude again from Lemma \ref{thm:TopMinors} that there is a finite number of such configurations at fixed genus. 
    \begin{figure}
        \centering
        \includegraphics[scale=.65]{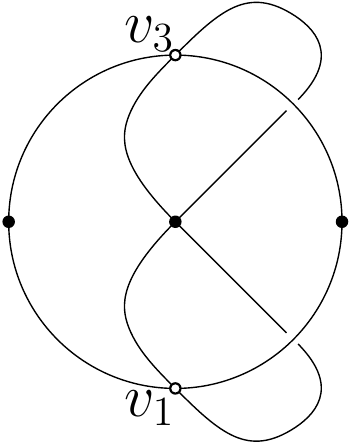}
        \hspace{1cm}
        \includegraphics[scale=0.45]{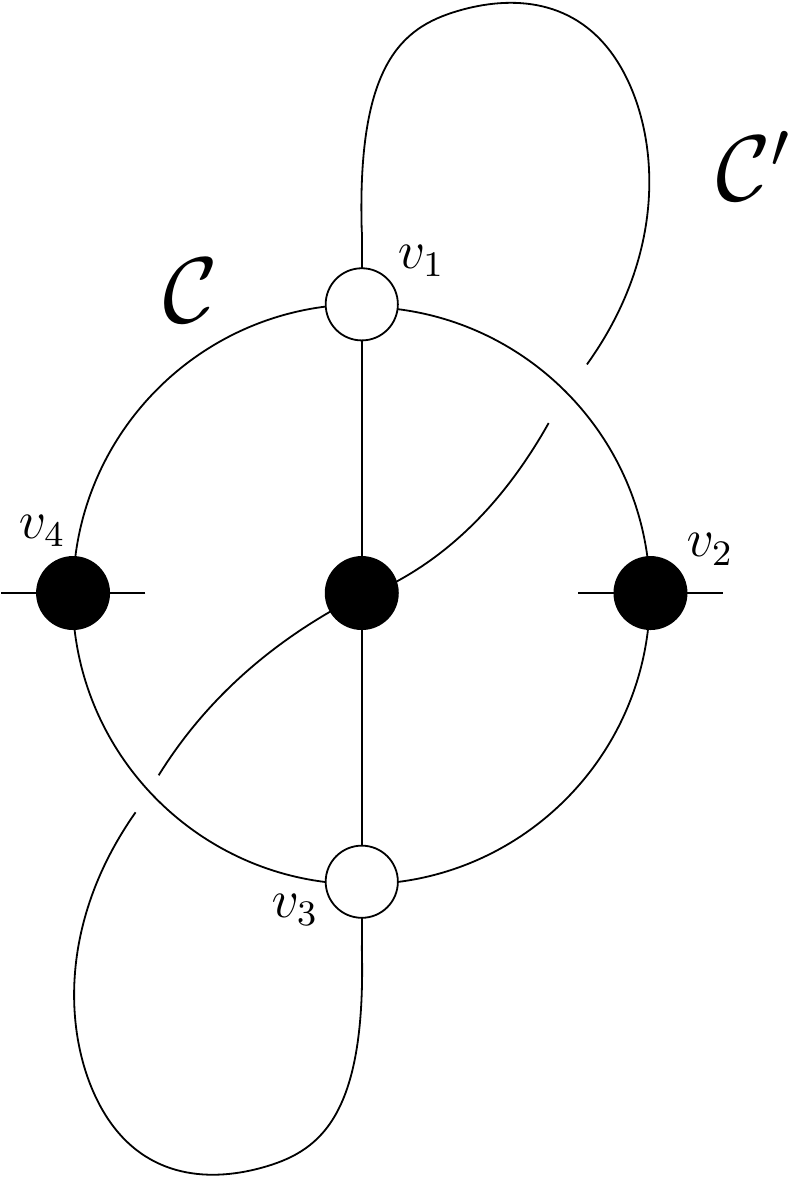}
        \caption{There are two possibilities for $\cC$ non-self-intersecting and $\cC'$ self-intersecting to intersect at two vertices of the same color, up to symmetry.}
        \label{fig:4nsi_with_si}
    \end{figure}
    \item The last situation, which we have not yet treated, is when $\cC'$ is non-self-intersecting. Then denote $v_{13}$ and $v_{31}$ the two other vertices of $\cC'$,
    \begin{equation} \label{LastCase}
    \includegraphics[scale=.65,valign=c]{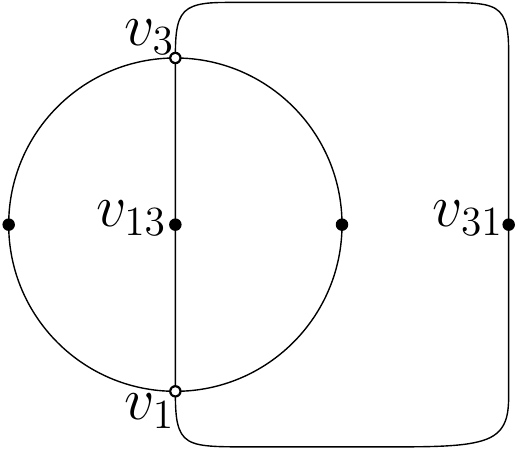} \hspace{1cm} \includegraphics[scale=.65,valign=c]{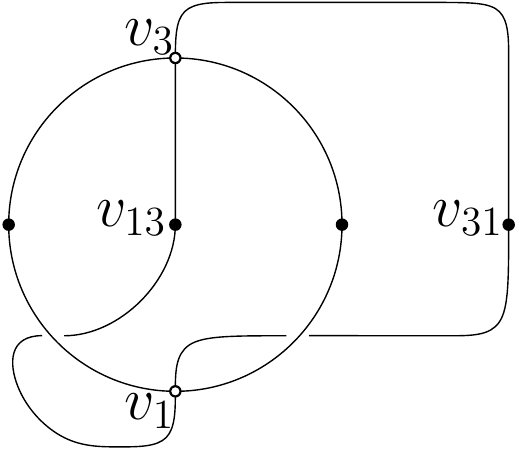}
    \end{equation}
    In the second case, $\cC\cup\cC'$ forms a topological minor of genus 1 and again from Lemma \ref{thm:TopMinors} there is a bounded number of such configurations.
    
    We thus focus on the first case of \eqref{LastCase} and consider the straight cycle $\cC''$ through $v_{13}$. If $\cC'\cup \cC''$ forms a configuration which we have already bounded, so is $\cC\cup \cC'\cup \cC''$. Therefore the only situation to be treated is when $\cC'\cup\cC''$ has the same structure as in the first case of \eqref{LastCase}, i.e. $\cC''$ goes through $v_{13}$ and $v_{31}$ and $\cC'\cup\cC''$ is planar. 
    This situation is represented on Figure~\ref{fig:4nsi_triple}.
    \begin{figure}[!h]
        \centering
        \includegraphics[scale=0.5]{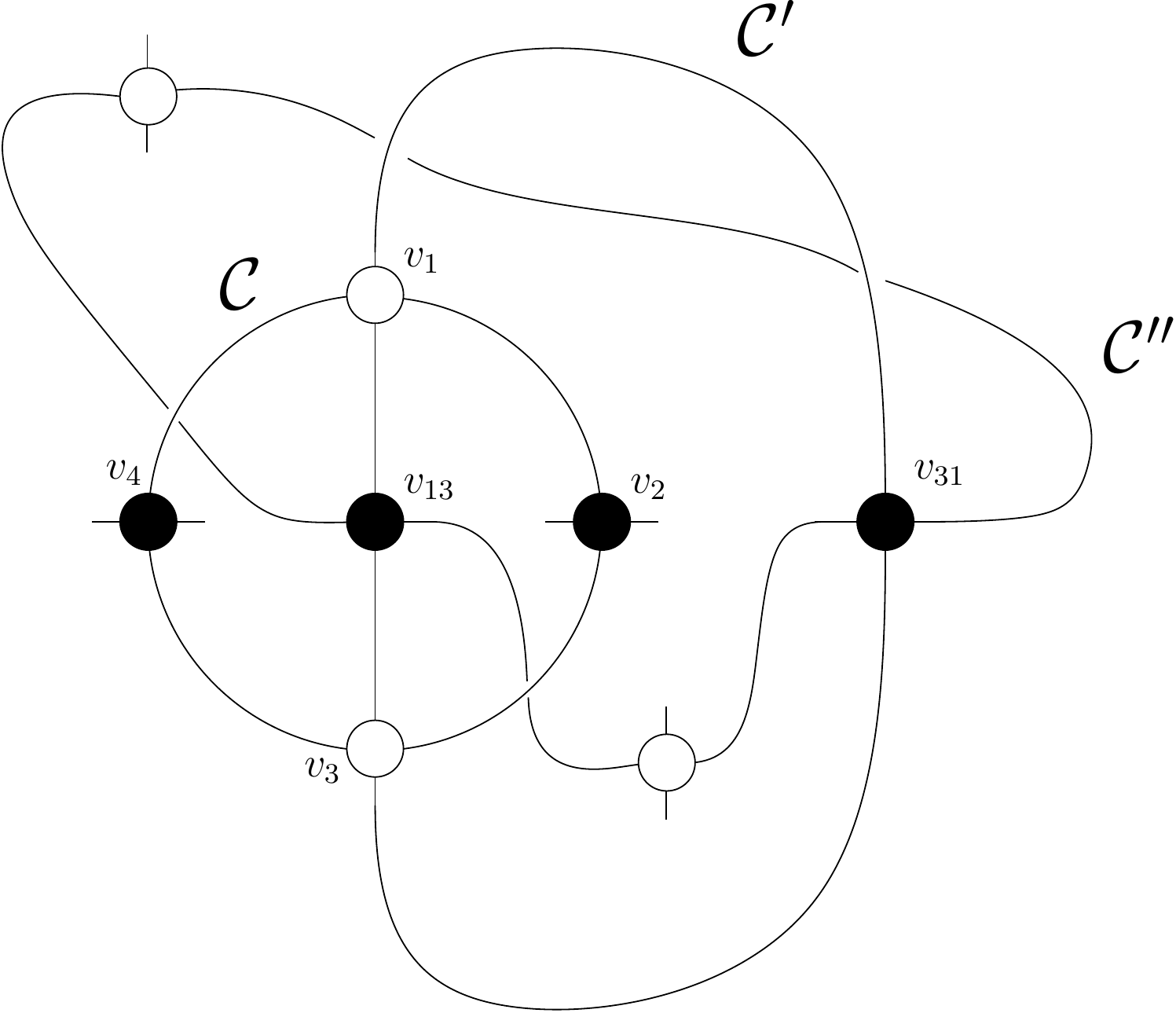}
        \caption{$\cC\cup\cC'\cup\cC''$ forms a topological minor of genus $2$.}
        \label{fig:4nsi_triple}
    \end{figure}
    In this configuration,  $\cC\cup\cC'\cup\cC''$ forms a topological minor of genus $2$, and it is concluded from Lemma \ref{thm:TopMinors} that there is a bounded number of such configurations.
\end{itemize}
We have exhausted all the possible cases for a straight cycle of length 4 and find that there is indeed a bound $\phi_{2n}\leq b_{2n}(g,l,k)$ for $n=2$. This was the last missing bound, so we have them for all $n\geq 1$, which proves Proposition \ref{prop:sch_ns} and as we have explained, Proposition \ref{prop:sch_ns} together with Proposition \ref{prop:sch_chain} proves Theorem \ref{thm:FiniteSchemesI}. 
\end{proof}

\subsection{\label{UNOD:dom_scheme} Identification of dominant schemes and double scaling limit}

\subsubsection{Structure of the dominant schemes} \label{UNOD:TreeStructure}

Via Theorem~\ref{thm:SchemesUNxOD}, we know that for given genus and grade $(g,l)$, there are finitely many schemes. Thus all singularities of the generating function for $\mathbb{MS}_{g,l}$ comes from the generating series of melons and chains. As we have shown, the only divergent objects at the critical point are the broken chains. Moreover, we restrict attention to schemes of vanishing grade, as proposed in \cite{BeCa}. The schemes of vanishing grade which have maximal number of broken chains for their genus are said to be \emph{dominant}. They are the schemes contributing at leading order in the double scaling limit as they are the one that "diverge the most" near critical point $t_c$.

Proposition~\ref{prop:sch_chain} gives the structure of the dominant schemes of genus $g$. All their chain-vertices are separating, i.e. they are schemes whose skeleton graphs are plane trees. In order to maximize the singularity, all chain-vertices are broken chain-vertices. Given Lemma \ref{thm:SkeletonGraph}, maximizing the number of chain-vertices is a linear problem on the degrees of the internal vertices of $\mathcal{I}(\cS)$ and their genus and grades, and on the genus and grades of the leaves. The solution is for $\mathcal{I}(\cS)$ to be a plane binary tree, whose internal vertices correspond to components of $\cS$ of vanishing genus and grade, and whose leaves correspond to components of genus 1 (and vanishing grade).

By repeating the analysis of \cite{BeCa} to identify the components of genus 1 and the components of vanishing grade and genus, the following proposition is obtained.

\begin{prop} \label{prop:dom_scheme_un}
A rooted dominant scheme of genus $g\geq1$ has $2g-1$ broken chain-vertices which are all separating. Such a scheme has the structure of a rooted\footnote{The root is a marked leaf.} binary plane tree where
\begin{itemize}
\item Each edge corresponds to a broken chain-vertex.
\item The root of the tree corresponds to the root of the map.
\item Each of the $g$ leaves is one of the following two graph
\begin{equation}
\includegraphics[scale=0.4]{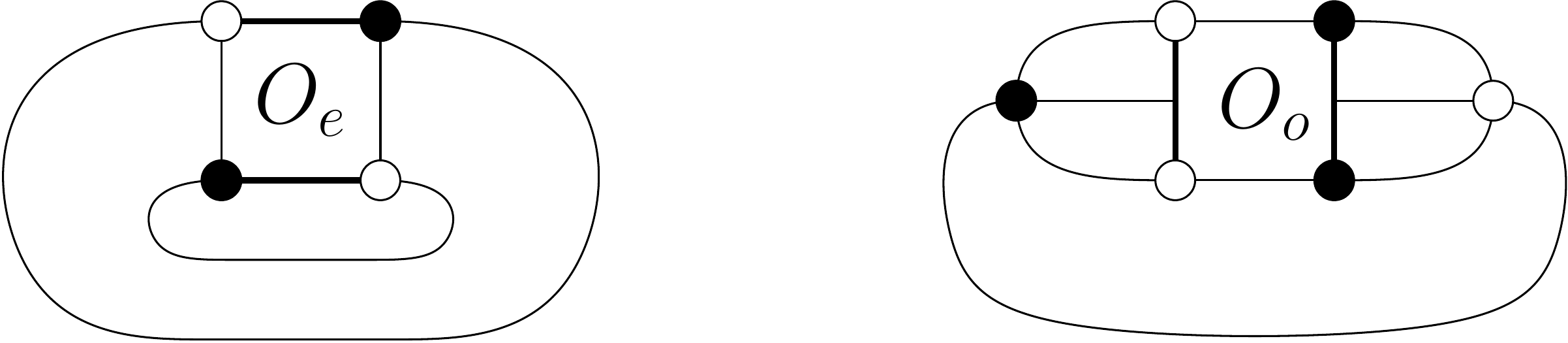}
\end{equation}
\item {Each internal vertex corresponds to one of the four $6$-point submaps
\begin{equation}
\includegraphics[scale=0.85]{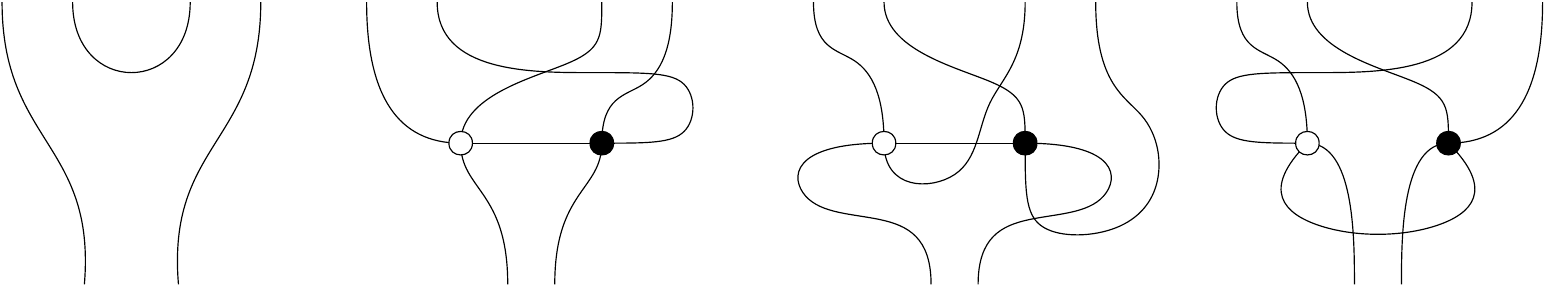}
\end{equation}
with the convention that each bottom left vertex is related to a black vertex.}
\end{itemize}
\end{prop}

From this proposition, it follows that the generating function of dominant schemes of genus $g$ is given by
\begin{align}
G_{dom}^g(t) = L^{2-2g}M(t) \Cat_{g-1} \left(C_{O,e}+t C_{O,o}\right)^g\left(1+6t\right)^{g-1} B(t)^{2g-1}
\end{align}
where we recall that $L=N/\sqrt{D}$ and a map of genus $g$ picks up the factor $L^{2-2g}$ from the large $N$, large $D$ expansion (see Theorem~\ref{thm:free_energy_U(N)}).

\subsubsection{Double scaling limit}

We recall that the critical point for $B(t)$ is $t_c = \frac{3^3}{4^4}$. Near this point we have:
\begin{align}
C_{O,e} = \frac{1}{8} + O(\frac{t}{t_c}) \\
C_{O,o} = \frac{1}{24} + O(\frac{t}{t_c}) \\
U(t) \underset{t \rightarrow t_c}{\sim} \frac{1}{3} - \sqrt{\frac{8}{27}}\left(1-\frac{t}{t_c}\right)^{-\frac{1}{2}} 
\end{align}
Therefore it follows that
\begin{align}
B(t) \underset{t \rightarrow t_c}{\sim} \sqrt{\frac{27}{8}}\left(1-\frac{t}{t_c}\right)^{-\frac{1}{2}}
\end{align}
Thus near the critical point $t_c$ we find
\begin{equation}
G_{\text{dom}}^g(t) \underset{t \rightarrow t_c}{\sim} L^{2-2g}\frac{4}{3}\Cat_{g-1} \left(\frac{1}{8}+t_c \frac{1}{24}\right)^g \left(1+6t_c\right)^{g-1} \left(\frac{27}{8}\frac{1}{\left(1-\frac{t}{t_c}\right)}\right)^{g-\frac{1}{2}}
\end{equation}
Therefore, after introducing the double scaling parameter $\kappa$ defined as
\begin{equation}
\kappa^{-1} = L^2 \frac{1}{\left(\frac{1}{8}+t_c \frac{1}{24}\right)\left(1+6t_c\right)}\frac{8}{27}\left(1-\frac{t}{t_c}\right)
\label{eq:kappa_un}
\end{equation}
the following equality holds
\begin{align}
\left(1+3t\right)^{-1}\left(\frac{8}{27}\left(1-\frac{t}{t_c}\right)\right)^{\frac{1}{2}} = \frac{\kappa^{-\frac{1}{2}}}{L} \left(\frac{\frac{1}{8}+t_c \frac{1}{24}}{1+6t_c}\right)^{\frac{1}{2}}
\end{align}
so that in the double scaling limit, $G_{dom}^g(t)$, for $g>0$, contributes to the 2-point function as
\begin{equation}
G_{dom}^g(t) = \frac{4}{3}L\Cat_{g-1} \left(\frac{\frac{1}{8}+t_c \frac{1}{24}}{1+6t_c}\right)^{\frac{1}{2}}\kappa^{g-\frac{1}{2}}.
\end{equation}
For $g=0$, we simply have to account for melons thus contributing as $M(t_c) = \frac{4}{3}$.

\vspace{10pt}
Summing over the genus $g$, it is found that the series is convergent for $\kappa \leq \frac{1}{4}$. This gives the expression of the 2-point function in the double scaling limit,
\begin{align}
G_2^{DS}(t) &= L^{-2} \sum\limits_{g \geq 0}^{} G_{dom}^g(t) \\
&= 4/3 + \frac{4}{3}\frac{1}{L\kappa^{\frac{1}{2}}} \left(\frac{\frac{1}{8}+t_c \frac{1}{24}}{1+6t_c}\right)^{\frac{1}{2}} \sum\limits_g \Cat_{g-1} \kappa^g \nonumber \\
&= 4/3 + \frac{2}{3}\frac{1}{L}\left(\frac{\frac{1}{8}+t_c \frac{1}{24}}{1+6t_c}\right)^{\frac{1}{2}} \frac{1-\sqrt{1-4\kappa}}{\kappa^\frac{1}{2}}
\end{align}
This function has a square-root singularity, just like the double scaling 2-point function of tensor graphs in \cite{GuSch}, as well as the one in the multi-orientable model \cite{GuTaYo} and the $O(N)^3$-invariant model \cite{Bonzom4}, and also similar to that of the $U(N)^2\times O(D)$-invariant model with tetrahedral interaction \cite{BeCa} and to the double scaling function we found in the previous model, Equation~\eqref{eq:GDS_U(N)2}.

\bigskip

\noindent
{\bf Acknowledgements.} The authors are partially supported by the ANR-20-CE48-0018 3DMaps grant. 
A. T. has been partially
supported by the PN 09370102 grant.
VB is partially supported by the ANR-21-CE48-0017 LambdaComb grant. VN is grateful to LIPN for its hospitality.

\bibliography{DS_MM}

\end{document}